%% file: main.tex
\documentclass[a4paper,reqno]{amsart}
\usepackage{amsfonts}
\usepackage{array,calc} 
\usepackage{pgf}
\usepackage{tikz}
\usetikzlibrary{shapes,arrows,automata,decorations.pathreplacing,angles,quotes,calc}
\usepackage{tkz-euclide}
\usepackage{amsmath,calc,graphicx}
\usepackage{url}
\usepackage{hyperref}
\usepackage{amssymb}
\usepackage{amsthm}
\usepackage{float}
\usepackage{enumitem}
\usepackage{amsrefs}
\numberwithin{equation}{section}
\numberwithin{figure}{section}
\theoremstyle{plain}
\newtheorem{theorem}{Theorem}[section]
\newtheorem{lemma}[theorem]{Lemma}
\newtheorem{corollary}[theorem]{Corollary}
\newtheorem{proposition}[theorem]{Proposition}
\newtheorem{definition}[theorem]{Definition}

\newtheorem{prob}[theorem]{Problem}
\theoremstyle{remark}
\newtheorem{remark}[theorem]{Remark}

\newtheorem*{lem*}{\textsc{Lemma}}
\newtheorem*{cor*}{\textsc{Corollary}}
\newtheorem*{exer*}{\textsc{Exercise}}
\newtheorem*{con*}{\textsc{Conjecture}}
\newtheorem*{thm*}{\textsc{Theorem}}

\newcommand{\beq}{\begin{equation}}
\newcommand{\eeq}{\end{equation}}

\newcommand\Pone{\textrm{P}_{\textrm{I}} }

\newcommand\Pfour{\textrm{P}_{\textrm{IV}}}

\newcommand\Psix{\textrm{P}_{\textrm{VI}}}

\newcommand\B[1]{\ensuremath{\scalebox{0.5}{$($}{#1}\scalebox{0.5}{$)$}}}

\newcommand{\orcidauthorA}{0000-0001-7504-4444}

\title{On the Riemann-Hilbert problem for a $q$-difference Painlev\'e equation}
\date{}
\author[Nalini Joshi]{Nalini Joshi}
\thanks{Corresponding author: Nalini Joshi. NJ's ORCID ID is \orcidauthorA. Her research was supported by an
  Australian Research Council Georgina Sweet Laureate Fellowship \#FL120100094 and Discovery Projects \#DP130100967 and \#DP200100210.}
\address{School of Mathematics and Statistics F07, The University of Sydney, NSW 2006, Australia}
\email{nalini.joshi@sydney.edu.au}
\author{Pieter Roffelsen}
\thanks{PR acknowledges the support of the H2020-MSCA-RISE-2017 PROJECT
  No. 778010 IPADEGAN.}
\address{SISSA, Trieste, Italy}

\begin{document}
\begin{abstract}
A Riemann-Hilbert problem for a $q$-difference Painlev\'e equation, known as $q\Pfour$, is shown to be solvable. This yields a bijective correspondence between the transcendental solutions of $q\Pfour$  and corresponding data on an associated $q$-monodromy surface. We also construct the moduli space of $q\Pfour$ explicitly.
\end{abstract}
\maketitle

\tableofcontents
\input{intro}
\input{prelim}

\input{fuchsian}
\input{isomonodromic}
\input{modulispace}

\input{conclusion}

\appendix
\input{birational}
\input{cubicderivation}

\end{document}

%% file: intro.tex
\section{Introduction}
In this paper, we consider the Riemann-Hilbert problem (RHP) arising from a linear $q$-difference equation
\begin{equation}\label{eq:0}
Y(q z)=A(z)Y(z),
\end{equation}
where $A$ is a $2\times 2$ matrix polynomial in $z$. In particular, we are concerned with the invertibility of the map $\mathcal R$ from coefficients of $A$ to the class of connection matrices relating two fundamental solutions of Equation \eqref{eq:0} and we prove this for a class of cubic matrices $A$: $A=A_0+A_1z+A_2 z^2+A_3 z^3$, which satisfy the conditions given in Definition \ref{pro:aa}.

This paper is motivated by the study of $q$-discrete Painlev\'e equations, which arise when the matrix $A$ in Equation \eqref{eq:0} is deformed by a parameter $\lambda$. When the deformation is compatible with Equation \eqref{eq:0}, certain coefficients of $A$ turn out to satisfy discrete Painlev\'e equations \cites{jimbosakai,sakai2001,kny:17,joshicbms2019} with independent variable $\lambda$. In this setting, invertibility of $\mathcal R$ and information about its corresponding codomain is essential for describing solutions of the associated discrete Painlev\'e equation.

The $q$-discrete version of the RHP arising from Equation \eqref{eq:0} has properties that closely parallel RHPs used to study solutions of the classical differential Painlev\'e equations. In the latter case, the linear problem corresponding to Equation \eqref{eq:0} is a matrix linear differential equation $Y_z=A_d(z)Y$, where $A_d$ is rational in $z$, with poles at a sequence of given points (including infinity). The monodromy group describing how the fundamental solutions change as $z$ moves around such points is an inherent part of the description of the corresponding RHP. The solution of the RHP in this context also provides a map from the coefficients of $A_d$ to the monodromy data, whose codomain is called the monodromy manifold or moduli space of the Painlev\'e equation \cites{putsaito2009,chekhovmazzoccorubtsov2017}.

For ease of exposition, we will use analogous terminology in this paper. That is, we call $z$ the monodromy variable of Equation \eqref{eq:0} and refer to the codomain of $\mathcal R$ as the monodromy manifold or moduli space of the equation. The connection matrices relating fundamental solutions of Equation \eqref{eq:0} play the role of the generators of the monodromy group and their invariance (modulo conjugation) under compatible deformation in $\lambda$ will be referred to as the isomonodromy property. We note that despite this analogy, the methodology associated with the Riemann-Hilbert method has not been extended to $q$-discrete Painlev\'e equations, although partial attempts exist in the literature \cites{mano2010,joshiroffelsen}. 

Modern developments in the mathematical study of RHPs  have their origin in the solution of the  Korteweg-de Vries equation \cite{ggkm67} through the inverse scattering method. Reductions of such PDEs led to the Painlev\'e equations \cites{JMMS,jimbomiwaI,jimbomiwaII,jimbomiwaIII}, where RHPs have played a major role in the deduction of asymptotic properties of solutions. In this setting, a key step was provided by a steepest descent method developed by Deift and Zhou \cites{deift93,fokas}. See \S \ref{s:bg} for further information about the background. 

The most general results known about the $q$-discrete RHP, associated with Equation \eqref{eq:0}, come from classical studies by Birkhoff and Carmichael \cites{carmichael1912,birkhoffgeneralized1913} with recent refinements by Sauloy \cite{sauloy2002}, under certain conditions on $A$, which we call \textit{Fuchsian} and \textit{non-resonant} below (see Definition \ref{def:qlaxcondns}). 

Two further elements of the classical study of Equation \eqref{eq:0} are worth noting here. The first concerns conditions on the boundary curve $\gamma$ separating domains of existence of the fundamental solutions. We describe these conditions in Definition \ref{def:contour} \cites{carmichael1912,birkhoffgeneralized1913}. The second arises from the fact that solutions of the associated $q$-discrete Painlev\'e equation are not bounded for all $\lambda\in\mathbb C$. For the case studied in this paper, such singular values are included by solving a sequence of RHPs (corresponding to a sequence of iterates of the $q$-discrete Painlev\'e equation).

\subsection{Main Results}\label{s:mr}
Our results hold for  Fuchsian, non-resonant linear $q$-difference equations defined below. Note that we assume $0\not=|q|<1$ thoughout the paper, without loss of generality, and furthermore, we use the notation $\mathbb C^*$ to refer to the complex plane punctured at the origin. Our main results are stated as Theorems \ref{thm:mainresult1}, \ref{thm:rhmap} and \ref{thm:mainresult3}.

\begin{definition}\label{def:qlaxcondns}
  Assume that the coefficient matrix of Equation \eqref{eq:0} is polynomial{\rm :} $A(z)=A_0+A_1 z+\ldots + A_n\,z^n$, for $0\not=n\in\mathbb N$ with $A_n\neq 0$. Let the eigenvalues of $A_0$ be $\theta_1, \theta_2$ and those of $A_n$ be $\kappa_1, \kappa_2$. Moreover, let the zeroes of $\det(A)$ be $\{x_1, \ldots, x_{2n}\}$.  The quantities $\theta_i$, $\kappa_i$, $i=1,2$ and $x_j$, $1\le j\le 2n$, are called {\em critical exponents} and the collection  $\{\theta_i, \kappa_i, x_j\}$ is called the {\em critical data} of \eqref{eq:A}. 
  \begin{enumerate}[leftmargin=0.8cm,label={\rm(\alph*)}]
  \item Equation \eqref{eq:0} is said to be {\em Fuchsian} if
   ${\rm det}(A_0)\not=0$ and ${\rm det}(A_n)\not=0$. 
  \item  In the case when Equation \eqref{eq:0} is Fuchsian, it is called {\em non-resonant} if the critical data satisfy the following conditions:
    \begin{enumerate}[leftmargin=1.2cm,label={\rm(\roman*)}]
      \item $\theta_1/\theta_2\not=q^m$, $\kappa_1/\kappa_2\not=q^m$, for any integer $m$.
        \item $x_i/x_j\not=q^m$, for $i, j=1,\ldots, 2n$, $i\not=j$ and any integer $m$.
        \end{enumerate}
        \end{enumerate}
    \end{definition}
Note that the critical data are related by
    \begin{equation}\label{eq:constraint}
\theta_1\theta_2=\kappa_1\kappa_2\prod_{1\leq i \leq 2n}{x_i}.
\end{equation}
We recall Birkhoff's definition of appropriate contours for the corresponding Riemann problem  \cite{birkhoffgeneralized1913}. 
\begin{definition}\label{def:contour}
        Given non-resonant $\{x_1, \ldots, x_{2n}\}$, which we assume to be invariant under negation, a positively oriented Jordan curve $\gamma$ in $\mathbb{C}^*$ is called \emph{admissible} if all the following conditions are satisfied:
     \begin{enumerate}[label={{\rm (\roman *)}}]
     \item It admits local parametrization by analytic functions at each point. 
     \item It possesses the reflection symmetry $\gamma=-\gamma$.
     \item Letting the region on the left {\rm (}respectively right{\rm )} of $\gamma$ in $\mathbb{C}$ be $D_-$ and $D_+$, we have
	\begin{equation}\label{eq:insideoutside}
	q^kx_i\in\begin{cases}
	D_- & \text{if $k>0$,}\\
	D_+ & \text{if $k\leq 0$,}
	\end{cases}
	\end{equation} 
	for all $k\in\mathbb{Z}$ and $1\leq i \leq 2n$.
      \end{enumerate}
\end{definition}
Equation \eqref{eq:0} has two fundamental solutions defined respectively in $D_\pm$. A relation between the two solutions is provided by a connection or jump matrix $C$. We assume below that $C$ is a $2\times 2$ connection matrix satisfying the first three properties of Definition \ref{def:connectionspace}.
The RHP then takes the following form.
\begin{definition}[Riemann-Hilbert problem]\label{def:RHP1and2}
  Given non-resonant critical exponents $\theta_{1,2}$, $\kappa_{1,2}$ and $x_k$, $1\leq k \leq 2n$, a $2\times 2$ connection matrix $C(z)$ satisfying properties \eqref{item:c1}, \eqref{item:c2}, and \eqref{item:c3} of  Definition \ref{def:connectionspace}, an admissible curve $\gamma$, and $m\in\mathbb{Z}$, a $2\times2$ complex matrix function $Y^{\B{m}}(z)$ is called a solution of the Riemann-Hilbert problem $\textnormal{RHP}^{\B{m}}(\gamma,C)$ if it satisfies the following conditions.
  \begin{enumerate}[label={{\rm (\roman *)}}]
  \item $Y^{\B{m}}(z)$ is analytic on $\mathbb{C}\setminus\gamma$.
    \item $Y^{\B{m}}(z')$ has continuous boundary values $Y_-^{\B{m}}(z)$ and $Y_+^{\B{m}}(z)$ as $z'$ approaches $z\in \gamma$ from $D_-$ and $D_+$ respectively, where
		\begin{equation}\label{eq:jumpm}
		Y_+^{\B{m}}(z)=Y_-^{\B{m}}(z)C(z),\quad z\in \gamma.
              \end{equation}
            \item $Y^{\B{m}}(z)$ satisfies
              \begin{equation}\label{eq:normalisationm}
		Y^{\B{m}}(z)=\left(I+\mathcal{O}\left(z^{-1}\right)\right)z^{m\sigma_3}\quad z\rightarrow \infty.
              \end{equation}
              \end{enumerate}
  \end{definition}
\noindent Here $\sigma_3$ is a Pauli spin matrix (see \S \ref{s:not}).

We assume $n=3$ in Theorems \ref{thm:mainresult1}, \ref{thm:rhmap} and \ref{thm:mainresult3}.  Theorem \ref{thm:mainresult1} shows that a solution of the RHP \eqref{def:RHP1and2} exists and relates the result to the coefficient matrix $A$, whose entries in turn are related to a $q$-discrete Painlev\'e equation. The mapping $\mathcal R$ is shown to be bijective in Theorem \ref{thm:rhmap}. Finally, the monodromy surface is described in Theorem \ref{thm:mainresult3}. These theorems are proved respectively in Sections \ref{subsection:timedef}, \ref{section:rhcorrespondence} and \ref{subsec:proof}.

\subsection{Outline of the paper}\label{s:out}
The linear $q$-difference problem of interest is given by \eqref{eq:laxspectral}. We describe its properties, as well as its associated discrete Painlev\'e equation, before stating our main results in Section \ref{s:pre}. In Section \ref{sec:class}, we analyse the corresponding direct and inverse monodromy problems. In Section \ref{sec:isomonodromic}, we show how $q\Pfour$ (i.e., Equation \eqref{eq:qp4}) defines an isomonodromic deformation of Equation \eqref{eq:laxspectral} and prove Theorems \ref{thm:mainresult1} and \ref{thm:rhmap}.
In Section \ref{section:moduli}, we study the monodromy surface and prove Theorem \ref{thm:mainresult3} and Remark \ref{rem:real}. The paper ends with a conclusion given in Section \ref{s:con}.

\subsection{Background}\label{s:bg}
The six classical differential Painlev\'e equations $\Pone - \Psix$, were identified more than a century ago, while the discrete Painlev\'e equations are a more recent discovery. We focus on the $q$-discrete Painlev\'e equations, which are iterated on spirals in the complex plane parametrized by $\lambda=\lambda_0\,q^n$, for some given complex $q\not=0,1$ and $\lambda_0\neq 0$. See \cites{sakai2001,kny:17,joshicbms2019}. 

Each Painlev\'e equation is a compatibility condition for a pair of associated linear problems called a Lax pair:
\begin{align*}
    Y_z(z,\lambda)&=A(z,\lambda)Y(z,\lambda),\\
    Y_\lambda(z,\lambda)&=B(z,\lambda)Y(z,\lambda).
\end{align*}
Here we adopt the convention that $\lambda$ denotes the independent variable of the associated Painlev\'e equation, while $z$ is a monodromy variable. 

The monodromy data describe the behaviour of a fundamental solution near each of the singularities of $A(z,\cdot)$ in $z\in\mathbb C$ and at $z=\infty$. Under variation of $\lambda$, the Painlev\'e flow deforms the linear system in such a way that the monodromy data are left invariant. Such monodromy data lie on explicitly defined affine cubic surfaces \cite{putsaito2009}. The latter are moduli spaces of the corresponding Painlev\'e equations. For example, the moduli space of the first Painlev\'e equation
\begin{equation*}
    y_{\lambda\lambda}=6y^2+\lambda
\end{equation*}
is given by the affine cubic surface \cites{putsaito2009,phdroffelsen,chekhovmazzoccorubtsov2017}
\begin{equation*}
    \{x\in\mathbb{C}^3:x_1x_2x_3+x_1+x_2+1=0\}.
\end{equation*}

Conversely, given prescribed monodromy data, corresponding to a point on such a surface, the inverse problem asks for the corresponding solution of a Painlev\'e equation. This problem can be recast into an RHP with suitable contours and jumps given in terms of the monodromy data. Deift and Zhou \cite{deift93} developed a method of steepest descent to analyse the solutions of RHPs, and this method has been extended to the Painlev\'e equations to provide global asymptotic information of their general solutions \cites{fokas92,fokas}. 

In the context of $q$-difference equations, the associated Lax pair no longer consists of differential equations, but instead becomes a pair of linear $q$-difference equations -- see Equations \eqref{eq:qLax}. The linear problem \eqref{eq:A} has singularities only at $z=0$ and $z=\infty$.
Under certain conditions (Definition \ref{def:qlaxcondns}), Carmichael \cite{carmichael1912} constructed fundamental solutions of Equation \eqref{eq:A} in neighbourhoods of each point and characterised the connection matrix relating them. The connection matrix embodies the monodromy data of this linear system. 

Given a connection matrix, Birkhoff \cite{birkhoffgeneralized1913} showed how the  problem of reconstructing a Fuchsian system with that connection matrix can be recast into a Riemann-Hilbert problem,  and proved that this inverse problem always has a solution.  A modern extension of this theory (to include non-Fuchsian cases) has also been developed by Ramis et al. \cite{sauloy}.

However, to the best of our knowledge, such a Riemann-Hilbert formulation has not been used to obtain information about general solutions of any $q$-discrete Painlev\'e equations, except in the case of $q\Psix$ \cites{mano2010,ormerodconnection,jimbonagoyasakai}. Analysis of such equations in certain limits has been carried out in two cases \cites{mano2010,joshiroffelsen,phdroffelsen}, but questions such as a bijection between the coefficients of the linear problem and the solutions of the nonlinear equation, or the moduli space of ``monodromy'' data have not been considered. This provides the motivation of our present paper.

\subsection{Notation}\label{s:not}
Define the Pauli matrices
\begin{equation*}
\sigma_1=\begin{pmatrix}
0 & 1\\
1 & 0
\end{pmatrix},\quad
\sigma_3=\begin{pmatrix}
1 & 0\\
0 & -1
\end{pmatrix}.
\end{equation*}
We define the $\mathit{q}$-Pochhammer symbol by means of the infinite product
\begin{equation*}
(z;q)_\infty=\prod_{k=0}^{\infty}{(1-q^kz)}\qquad (z\in\mathbb{C}),
\end{equation*}
which converges locally uniformly in $z$ on $\mathbb{C}$. In particular $(z;q)_\infty$ is an entire function, satisfying
\begin{equation*}
(qz;q)_\infty=\frac{1}{1-z}(z;q)_\infty,
\end{equation*}
with $(0;q)_\infty=1$ and simple zeros on the semi $q$-spiral $q^{-\mathbb{N}}$. The $\mathit{q}$-theta function is defined as
\begin{equation}\label{eq:thetasym}
\theta_q(z)=(z;q)_\infty(q/z;q)_\infty\qquad (z\in \mathbb{C}^*),
\end{equation}
which is analytic on $\mathbb{C}^*$, with essential singularities at $z=0$ and $z=\infty$ and simple zeros on the $q$-spiral $q^\mathbb{Z}$. It satisfies
\begin{equation*}
\theta_q(qz)=-\frac{1}{z}\theta_q(z)=\theta_q(1/z).
\end{equation*}

For $n\in\mathbb{N}^*$ we denote
\begin{align*}
\theta_q(z_1,\ldots,z_n)&=\theta_q(z_1)\cdot \ldots\cdot \theta_q(z_n),\\
(z_1,\ldots,z_n;q)_\infty&=(z_1;q)_\infty\cdot\ldots\cdot (z_n;q)_\infty.
\end{align*}

For conciseness, we will use bars to denote iteration in $\lambda$. That is, for $f=f(\lambda)$, we denote $f(q\,\lambda)=\overline f$, and $f(\lambda/q)=\underline f$.

We denote the complex projective space $\mathbb C\mathbb P$ by $\mathbb P$. Where a real projective space is needed, we distinguish it from $\mathbb P$ by using the notation $\mathbb P_{\mathbb R}$. In each case, we will work with the three-fold Cartesian product denoted by $\mathbb P^1\times \mathbb P^1\times \mathbb P^1$ or $\mathbb P^1_{\mathbb R}\times \mathbb P^1_{\mathbb R}\times \mathbb P^1_{\mathbb R}$.

%% file: prelim.tex
\section{Statement of the Results}\label{s:pre}
We consider $q$-difference Lax pairs of the form
\begin{subequations}\label{eq:qLax}
  \begin{align}
  Y(q z, \lambda)&=A(z, \lambda)Y(z,\lambda),\label{eq:A}\\
  Y(z, q \lambda)&=B(z, \lambda)Y(z,\lambda),\label{eq:B}
  \end{align}
\end{subequations}
where $A$ and $B$ are $2\times 2$ matrices polynomial in $z$ and $A$ satisfies the conditions in Definition \ref{def:qlaxcondns}. We assume that the system is compatible. That is, Equation \eqref{eq:A}$\bigm |_{\lambda\mapsto q\lambda}$ is equivalent to Equation \eqref{eq:B}$\bigm |_{z\mapsto qz}$, which requires
\begin{equation}\label{eq:ABcomp}
  A(z, q\lambda)B(z, \lambda)=B(qz,\lambda)A(z,\lambda).
  \end{equation}
  Equation \eqref{eq:A} is the key object of our study.

Suppose $A$ is non-resonant, Fuchsian and polynomial of degree 3 in $z$. Given $B$ satisfying Equation \eqref{eq:ABcomp}, $A$ and $B$ can be parametrized as follows \cite{joshinobu2016}:
\begin{subequations}\label{eq:JNLax}
\begin{align}
\notag A={\mathcal A}:=&\begin{pmatrix}
u & 0\\
0 & 1
\end{pmatrix}\begin{pmatrix}
-i\,q\frac{\lambda}{f_2}z & 1\\
-1 & -\,i\,q\frac{f_2}{\lambda} z
\end{pmatrix}
\begin{pmatrix}
-\,i\,a_0a_2\frac{\lambda }{f_0}z & 1\\
-1 & -\,i\,a_0a_2\frac{f_0}{\lambda} z
\end{pmatrix}\times\\
&\ \times\,
\begin{pmatrix}
-\,i\,a_0\frac{\lambda}{f_1}z & 1\\
-1 & -\,i\,a_0\frac{f_1}{\lambda} z
\end{pmatrix}\begin{pmatrix}
u^{-1} & 0\\
0 & 1
\end{pmatrix},\label{eq:AJN}\\
B={\mathcal B}:=&\begin{pmatrix}
0 & -bu\\
b^{-1}u^{-1} & 0
\end{pmatrix}+
 \begin{pmatrix}
z & 0\\
0 & 0
\end{pmatrix}\,,\label{eq:BJN}
\end{align}
\end{subequations}
where
\[
b=\frac{\lambda(1+a_1f_1(1+a_2f_2))}{i\,(q\lambda^2-1)f_2},
\]
and $f_i$, $i=0, 1, 2$ ,and $u$ are functions of $\lambda$, independent of $z$. 
These matrices correspond to the equations:
\begin{subequations}
	\label{eq:laxpair}
	\begin{align}\label{eq:laxspectral}
	Y(qz,\lambda)&={\mathcal A}(z;\lambda,f,u)Y(z,\lambda),\\
	Y(z,q\lambda)&={\mathcal B}(z;\lambda,f,u)Y(z,\lambda).\label{eq:laxtime}
	\end{align}
      \end{subequations}

      The compatibility condition is then a polynomial equation of degree 4 in $z$. An overdetermined system of  equations arises from the vanishing of the coefficient of each monomial $z^j$, $1\le j\le 4$, identically in $\lambda$. The system is satisfied if and only if the following equation holds:
\begin{equation}\label{eq:qp4}
\ q\Pfour(a):\begin{cases}\displaystyle
\frac{\overline{f}_0}{a_0a_1f_1}=\frac{1+a_2f_2(1+a_0f_0)}{1+a_0f_0(1+a_1f_1)}, &\\
\displaystyle\frac{\overline{f}_1}{a_1a_2f_2}=\frac{1+a_0f_0(1+a_1f_1)}{1+a_1f_1(1+a_2f_2)}, &\\
\displaystyle\frac{\overline{f}_2}{a_2a_0f_0}=\frac{1+a_1f_1(1+a_2f_2)}{1+a_2f_2(1+a_0f_0)}, &
\end{cases}
\end{equation}
where $f=(f_0,f_1,f_2)$ is a function of $\lambda$, $a:=(a_0,a_1,a_2)$ are constant parameters, subject to
\begin{equation}\label{eq:conditionfa}
f_0f_1f_2=\lambda^2,\quad a_0a_1a_2=q,
\end{equation}
and $u$ satisfies
\begin{equation}\label{eq:auxiliaryuf}
	\frac{\overline{u}}{u}=\left[\frac{\lambda(1+a_1f_1(1+a_2f_2))}{i(q\lambda^2-1)f_2}\right]^2.
      \end{equation}
 Equation \eqref{eq:qp4} is a $q$-discrete Painlev\'e equation often referred to as $q\Pfour$, because it has a continuum limit to the fourth Painlev\'e equation.\footnote{This
  equation has alternative names in the literature and is 
 also referred to as $q\Pfour(A_5^{(1)})$, for its initial
 value space, or $q\Pfour\bigl((A_2+A_1)^{(1)}\bigr)$, for its
 symmetry group -- see \cite{joshicbms2019}.}

\begin{remark}
          The variable $u$ did not occur in the original definition of the Lax pair in \cite{joshinobu2016}, but we have introduced it here for convenience. It arises from a gauge freedom by constant diagonal matrices. 
        \end{remark}
        
By carrying out the matrix multiplication, it is evident that ${\mathcal A}$ takes the form
\begin{align}\label{eq:Aexplicit}
{\mathcal A}(z;\lambda,f,u)= &\begin{pmatrix}
	0 & -u\\u^{-1} & 0
	\end{pmatrix}+
	z\begin{pmatrix}
	ig_1\lambda & 0\\0 & ig_2\lambda^{-1}
	\end{pmatrix}\\
	&+z^2\begin{pmatrix}
	0 & -ug_3\\u^{-1}g_4 & 0
	\end{pmatrix}+z^3 qa_0^2a_2i\begin{pmatrix}
	\lambda & 0\\
	0 & \lambda^{-1}
	\end{pmatrix},\nonumber
	\end{align}
	where $g=(g_1,g_2,g_3,g_4)$ satisfies the algebraic equations
\begin{subequations}\label{eq:algebraic}
	\begin{align}
	&g_1g_2-(g_3+g_4)=a_0^2(1+a_1^2a_2^2+a_2^2),\\
	&g_3g_4-qa_0^2a_2\left(g_1+g_2\right)=a_0^4a_2^2(1+a_1^2a_2^2+a_1^2).
	\end{align}
      \end{subequations}
      See Equation \eqref{eq:aginf} for the transformation from $f$ to $g$, and Equation \eqref{eq:afing} for its inverse. The algebraic surface in $\mathbb C^4$ defined by equations \eqref{eq:algebraic} will be denoted by $\mathcal{G}(a)$.

In the variable $g$, the compatibility condition \eqref{eq:ABcomp} is equivalent to the $q$-difference system:
\begin{equation}\label{eq:aqp4}
q\Pfour^{\text{mod}}(a):\begin{cases}
\overline{g}_1=q^{-1}\lambda^{-2}g_2+a_0^2a_2\lambda^{-2}g_3^{-1}(q\lambda^2-1),&\\
\overline{g}_2=q\lambda^2g_1-qa_0^2a_2g_3^{-1}(q\lambda^2-1),&\\
\overline{g}_3=g_4+\lambda^{-2}g_3^{-1}(q\lambda^2g_1-g_2)a_0^2a_2(q\lambda^2-1)&\\
\qquad -qa_0^4a_2^2\lambda^{-2}g_3^{-2}(q\lambda^2-1)^2,&\\
\overline{g}_4=g_3,&
\end{cases}
\end{equation}
on $\mathcal{G}(a)$. Note that while this system is apparently singular at $g_3=0$, the system is well-posed for neighbouring initial values in an annular region around the origin, punctured on the line $g_3=0$. Denote this annular region by $\mathcal D_0$. We denote the line $g_3=0$ by $\mathcal S$.

The iteration of initial values in $\mathcal D_0$ is well-defined and the corresponding orbit after 3 steps is continuous on the whole domain $\mathcal D_0\cup \mathcal S$. That is, $\overline{\overline{\overline{g}}}_i$, $1\le i\le 3$ (and similarly backward iterates) are well defined for a domain of initial values including $g_3=0$ in its interior. (This is part of a property called {\em singularity confinement} in the literature. See Equation \eqref{eq:gcontinuation} in Appendix \ref{appendix:singularityconfinement} for further detail.)

We define notation that incorporates such singularities below.

\begin{figure}[b]
	\centering
	\begin{tikzpicture}[scale=0.6]
	\draw[->] (-6,0)--(6,0) node[right]{$\Re{z}$};
	\draw[->] (0,-6)--(0,6) node[above]{$\Im{z}$};
	\tikzstyle{star}  = [circle, minimum width=3.5pt, fill, inner sep=0pt];
	\tikzstyle{starsmall}  = [circle, minimum width=3.5pt, fill, inner sep=0pt];
	
	\draw[domain=-1.5:6,smooth,variable=\x,red] plot ({exp(-\x*ln(2))*2*cos((\x*pi/16) r)},{exp(-\x*ln(2))*2*sin((\x*pi/16) r)});
	\draw[domain=-1.5:6,smooth,variable=\x,red] plot ({-exp(-\x*ln(2))*2*cos((\x*pi/16) r)},{-exp(-\x*ln(2))*2*sin((\x*pi/16) r)});
	
	\draw[domain=-1.1:6,smooth,variable=\x,red] plot ({exp(-\x*ln(2))*2*5/3*cos((3*pi/4+\x*pi/16) r)},{exp(-\x*ln(2))*2*5/3*sin((3*pi/4+\x*pi/16) r)});	
	\draw[domain=-1.1:6,smooth,variable=\x,red] plot ({-exp(-\x*ln(2))*2*5/3*cos((3*pi/4+\x*pi/16) r)},{-exp(-\x*ln(2))*2*5/3*sin((3*pi/4+\x*pi/16) r)});
	
	\draw[domain=-1.2:6,smooth,variable=\x,red] plot ({exp(-\x*ln(2))*2*4/3*cos((pi/4+\x*pi/16) r)},{exp(-\x*ln(2))*2*4/3*sin((pi/4+\x*pi/16) r)});	
	\draw[domain=-1.2:6,smooth,variable=\x,red] plot ({-exp(-\x*ln(2))*2*4/3*cos((pi/4+\x*pi/16) r)},{-exp(-\x*ln(2))*2*4/3*sin((pi/4+\x*pi/16) r)});
	
	\draw [blue,thick,decoration={markings, mark=at position 0.68 with {\arrow{<}}},
	postaction={decorate}] plot [smooth cycle,tension=0.6] coordinates {(3,0) (3.2,-4) (1,-4.4) (-3,-3) (-3,0) (-3.2,4) (-1,4.4) (3,3) };	
	
	\node[starsmall]     (or) at ({0},{0} ) {};
	\node     at ($(or)+(-0.3,0.3)$) {$0$};
	
	\node[starsmall]     (qa0) at ({2*4/3*cos(pi/4 r)},{2*4/3*sin(pi/4 r)} ) {};
	\node     at ($(qa0)+(-0.3,0.3)$) {$qx_1$};
	\node[star]     (a0) at ({2*2*4/3*cos((pi/4-pi/16) r)},{2*2*4/3*sin((pi/4-pi/16) r)} ) {};
	\node     at ($(a0)+(-0.3,0.3)$) {$x_1$};
	
	\node[star]     (a1) at ({2*5/3*cos(3*pi/4 r)},{2*5/3*sin(3*pi/4 r)} ) {};
	\node     at ($(a1)+(0.45,0.25)$) {$q x_2$};
	\node[star]     (qma1) at ({2*2*5/3*cos((3*pi/4-pi/16) r)},{2*2*5/3*sin((3*pi/4-pi/16) r)} ) {};
	\node     at ($(qma1)+(0.4,0.25)$) {$x_2$};
	
	\node[star]     (q0) at ({2*1},{0} ) {};
	\node     at ($(q0)+(0.1,0.4)$) {$q x_3$};
	\node[star]     (qm) at ({2*2*cos((0-pi/16) r)},{2*2*sin((0-pi/16) r)} ) {};
	\node     at ($(qm)+(0.35,0.45)$) {$x_3$};
	
	\node[star]     (mqa0) at ({-2*4/3*cos(pi/4 r)},{-2*4/3*sin(pi/4 r)} ) {};
	\node     at ($(mqa0)+(-0.5,0.4)$) {$q x_4$};
	\node[star]     (ma0) at ({-2*2*4/3*cos((pi/4-pi/16) r)},{-2*2*4/3*sin((pi/4-pi/16) r)} ) {};
	\node     at ($(ma0)+(-0.4,0.4)$) {$x_4$};
	
	\node[star]     (ma1) at ({-2*5/3*cos(3*pi/4 r)},{-2*5/3*sin(3*pi/4 r)} ) {};
	\node     at ($(ma1)+(0.1,0.4)$) {$q x_5$};
	\node[star]     (mqma1) at ({-2*2*5/3*cos((3*pi/4-pi/16) r)},{-2*2*5/3*sin((3*pi/4-pi/16) r)} ) {};
	\node     at ($(mqma1)+(0.12,0.4)$) {$x_5$};
	
	\node[star]     (mq0) at ({-2*1},{0} ) {};
	\node     at ($(mq0)+(-0.1,0.3)$) {$q x_6$};
	\node[star]     (mqm) at ({-2*2*cos((0-pi/16) r)},{-2*2*sin((0-pi/16) r)} ) {};	
	\node     at ($(mqm)+(0,0.5)$) {$x_6$};
	
	\node[blue]     at (1.56,4.1) {$\boldsymbol{\gamma}$};
	
	\node[blue]     at (1.90,5.08) {$\boldsymbol{D_+}$};
	
	\node[blue]     at (1.28,2.75) {$\boldsymbol{D_-}$};
	
	\end{tikzpicture}
	\caption{Example of an admissable contour $\gamma$ in Definition \ref{def:contour}, where the six red spirals are $q^{\mathbb{R}}\cdot x_i$, $1\leq i\leq 6$.}
	\label{fig:contour}
\end{figure}
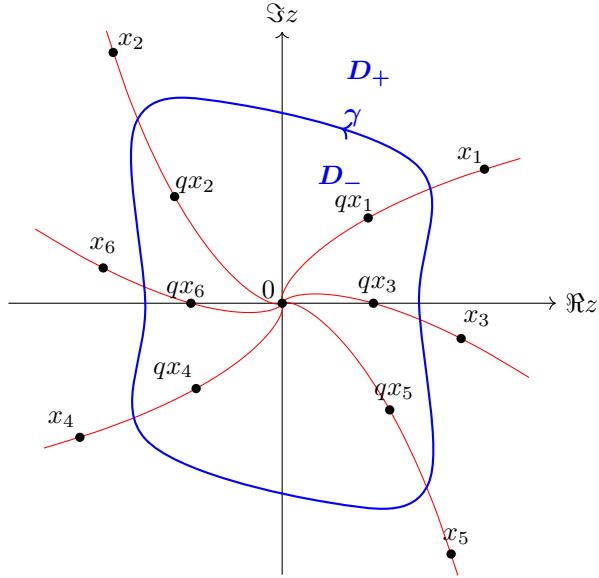

\begin{definition}\label{def:transcendent}
 Let $\lambda_0\in\mathbb{C}^*$ and $a\in\mathbb{C}^3$ be such that $\lambda_0^2\notin q^\mathbb{Z}$, with $a_0a_1a_2=q$. We call a sequence $(g^{\B{m}})_{m\in\mathbb{Z}}$ {\em a solution} of $q\Pfour^{\text{mod}}(\lambda_0,a)$ if 
 \begin{enumerate}[leftmargin=1.2cm,label={\rm(\roman*)}]
     \item  it satisfies Equation \eqref{eq:aqp4} with $\lambda=q^m\lambda_0$, $g_3^{\B{m}}\not=0$; or
     \item it satisfies the continuation Equations \eqref{eq:gcontinuation}, for $\lambda=q^m\lambda_0$, with $g_3^{\B{m-1}}=0$ or $g_3^{\B{m-2}}=0$, in which case we write $g^{\B{m}}=\mathsf{s}$. 
 \end{enumerate} 
 The solutions take values in $\mathcal{G}(a)\cup \{\mathsf{s}\}$ and we refer to them as  $q\Pfour^{\text{mod}}(\lambda_0,a)$--{\em transcendents}.  Analogous notions are defined for solutions of $q\Pfour(\lambda_0,a)$ by means of the birational equivalence given in \eqref{eq:aginf} and \eqref{eq:afing}.
\end{definition}

 Carmichael \cite{carmichael1912} constructed fundamental solutions of non-resonant Fuchsian systems in two domains, one with $0$ in its interior and another with $\infty$ in its interior.
 The critical exponents of the linear problem \eqref{eq:laxspectral}, as defined in Definition \ref{def:qlaxcondns}, are given by  
	\begin{equation}
	 \begin{aligned}
	x_1&=+a_0^{-1}, & x_2&=+a_1/q, & x_3&=+q^{-1},\\
	x_4&=-a_0^{-1}, & x_5&=-a_1/q, & x_6&=-q^{-1},
	\end{aligned}\label{eq:x1x6}
	\end{equation}
and $\theta_1=i$, $\theta_2=-i$, $\kappa_1=i q a_0^2a_2\lambda$, $\kappa_2=i q a_0^2a_2/\lambda$. So, by Definition \ref{def:qlaxcondns}, Equation \eqref{eq:laxspectral} is Fuchsian and it is non-resonant if and only if 
\begin{equation}\label{eq:conditionnonresonant}
\lambda^2,\pm a_0,\pm a_1,\pm a_2\notin q^\mathbb{Z}.
\end{equation}

 Birkhoff \cite{birkhoffgeneralized1913} defined admissible contours bounding such domains - see Definition \ref{def:contour}.  Figure \ref{fig:contour} illustrates such an adimissible contour $\gamma$ for our Equation \eqref{eq:laxspectral}.

Two fundamental solutions of the same linear system are related by a connection matrix. We collect the properties of such connection matrices in the following definition.
 \begin{definition}[Carmichael \cite{carmichael1912}]\label{def:connectionspace}
  Suppose $C(z)$ is a $2\times 2$ matrix function with the following properties:
  \begin{enumerate}[label={{\rm (c.\arabic*)}},ref=c.\arabic*]
		\item $C(z)$ is analytic in $\mathbb{C}^*$;\label{item:c1}
		\item $C(qz)=\frac{1}{qa_0^2a_2}z^{-3}\sigma_3C(z)\lambda^{-\sigma_3}$;\label{item:c2}
		\item $|C(z)|=c\theta_q(a_0z,-a_0z,a_0a_2z,-a_0a_2z,qz,-qz)$, for some $c\in\mathbb{C}^*$;\label{item:c3}
		\item $C(-z)=-\sigma_1C(z)\sigma_3$.\label{item:c4}
                \end{enumerate}
We define $\mathfrak{C}(\lambda,a)$ to be the set of all such functions. 
\end{definition}

We are now in a position to state the first of our main results, which provides the existence of a solution of $\textnormal{RHP}^{\B{m}}(\gamma,C)$. Note that for $m\in\mathbb{Z}$, if a solution $Y^{\B{m}}(z)$ of $\textnormal{RHP}^{\B{m}}(\gamma,C)$ exists, then Lemma \ref{lem:rhuniqueness} shows that it must be unique. The theorem below also shows how $Y^{\B{m}}(z)$ can be used to reconstruct a matrix polynomial $\mathcal A$ of degree 3 and how the entries of such a matrix solve $q\Pfour^{\text{mod}}(\lambda_0,a)$.

\begin{theorem}\label{thm:mainresult1} Let $\lambda_0\in\mathbb{C}^*$, $a\in\mathbb{C}^3$, $a=(a_0, a_1, a_2)$ be such that the non-resonance conditions \eqref{eq:conditionnonresonant} hold and let a matrix $C(z)\in \mathfrak{C}(\lambda_0,a)$ and an admissible curve $\gamma$ be given. Then the following results hold.
\begin{enumerate}[label={{\rm (\roman *)}},leftmargin=*]
\item For all $n\in\mathbb{Z}$, there exists at least one $m\in\{n,n+1,n+2\}$ such that a solution $Y^{\B{m}}(z)$ of $\textnormal{RHP}^{\B{m}}(\gamma,C)$ exists.
Define
  \[\mathfrak M=\{m \in \mathbb{Z} \Bigm | Y^{\B{m}}(z)\ \textrm{does not exist}\} \subseteq\mathbb{Z}.\] Then, $m\in \mathfrak M$, implies either $m+1\in \mathfrak M$ or $m-1\in \mathfrak M$.
		\item For $m\in\mathbb{Z}\setminus \mathfrak M$, let
		\begin{equation*}
		A^{\B{m}}(z):=\begin{cases}
		iqa_0^2a_2z^3Y^{\B{m}}(qz)\begin{pmatrix}
		\lambda_0 & 0\\
		0 & \lambda_0^{-1}
		\end{pmatrix} Y^{\B{m}}(z)^{-1}, &\text{if $z\in q^{-1}(D_+\cup \gamma)$,}\\
		Y^{\B{m}}(qz)\begin{pmatrix}
		i & 0\\
		0 & -i
		\end{pmatrix}C(z) Y^{\B{m}}(z)^{-1}, &\text{if $z\in D_+\cap q^{-1}D_-$,}\\
		Y^{\B{m}}(qz)\begin{pmatrix}
		i & 0\\
		0 & -i
		\end{pmatrix} Y^{\B{m}}(z)^{-1}. &\text{if $z\in D_-\cup \gamma$,}
		\end{cases}
		\end{equation*}
		Then $A^{\B{m}}(z)$ is a matrix polynomial of degree 3 in $z$ such that
			\begin{equation*}
			A^{\B{m}}(z)=\mathcal{A}(z;q^m\lambda_0,g^{\B{m}},u_m),
			\end{equation*}
			 for a unique $g^{\B{m}}\in\mathbb{C}^4$ satisfying \eqref{eq:algebraic} and $u_m\in\mathbb{C}^*$.
		\item Setting $g^{\B{m}}=\mathsf{s}$ and $u_m=0$ for $m\in \mathfrak M$, then the sequence $(g^{\B{m}})_{m\in\mathbb{Z}}$ is a solution of $q\Pfour^{\text{mod}}(\lambda_0,a)$ and $(u_m)_{m\in\mathbb{Z}}$ is a solution of the auxiliary equation
		\begin{subequations}\label{eq:auxiliary}
			\begin{align}
			\frac{u_{m+1}}{u_m}=&-\frac{1}{q^2a_0^4a_2^2}(q^m\lambda_0-q^{-m-1}\lambda_0^{-1})^{-2}(g_3^{\B{m}})^2,  &\text{if $g_3^{\B{m}}\neq 0$,}\label{eq:ub}\\
			\frac{u_{m+3}}{u_m}=&-\frac{1}{q^2a_0^4a_2^2}(q^m\lambda_0-q^{-m-3}\lambda_0^{-1})^{-2},   &\text{if $g_3^{\B{m}}=0$},\label{eq:ubbb}\\		u_{m}=&0,&\text{if $g^{\B{m}}=\mathsf{s}$}.
			\end{align}
		\end{subequations}
		This equation is birationally equivalent to \eqref{eq:auxiliaryuf}. In particular, for every $m_0\in\mathbb{Z}$, the Riemann-Hilbert problem $\textnormal{RHP}^{\B{m}}(\gamma,C)$ fails to have a solution at $m=m_0$ if and only if $g^{\B{m}}$ is singular \textnormal{(}i.e. $g^{\B{m}}=\mathsf{s}$\textnormal{)} at $m=m_0$.
	\end{enumerate}
\end{theorem}

\subsection{Injectivity}
Given an appropriate connection matrix $C(z)$, Theorem \ref{thm:mainresult1} associates to it a solution $g$ of $q\Pfour^{\text{mod}}(\lambda_0,a)$. However this association is not injective, due to the invariance of $g$ under right-multiplication of $C(z)$ by diagonal matrices. To overcome this issue, we define a quotient space of connection matrices.

\begin{definition}\label{def:monodromysurface}
	Let $\mathfrak{C}(\lambda_0,a)$ be as defined in Definition \ref{def:connectionspace} and let its quotient under right multiplication by diagonal matrices be denoted by $\mathfrak{C}(\lambda_0,a)/\operatorname{diag}(\mathbb{C}^*,\mathbb{C}^*)$. Then we define the the monodromy surface to be the space of equivalence classes under this quotient:
	\begin{equation*}
	\mathcal{M}_c(\lambda_0,a)=\mathfrak{C}(\lambda_0,a)/\operatorname{diag}(\mathbb{C}^*,\mathbb{C}^*),
      \end{equation*}
      and the Riemann-Hilbert mapping 
	\begin{equation}\label{eq:RHmapping}
          \mathcal{R}:\mathcal{M}_c(\lambda_0,a)\rightarrow
          \{  q\Pfour^{\text{mod}}(\lambda_0,a) \},
	\end{equation}
	by the action that assigns a unique $q\Pfour^{\text{mod}}(\lambda_0,a)$-transcendent to any equivalence class in $\mathcal{M}_c(\lambda_0,a)$ via Theorem \ref{thm:mainresult1}. 
\end{definition}
The Riemann-Hilbert mapping is shown to be bijective in the following theorem.
\begin{theorem}\label{thm:rhmap}Let $\lambda_0\in\mathbb{C}^*$ and $a\in\mathbb{C}^3$, satisfying $a_0a_1a_2=q$, such that the non-resonant conditions \eqref{eq:conditionnonresonant} are satisfied.
	Then the Riemann-Hilbert mapping \eqref{eq:RHmapping}
	is a bijection.
\end{theorem}

To define coordinates on the monodromy surface, we use the following notation: for any nonzero $2\times2$ matrix $R$ which is not invertible, let $r_1$ and $r_2$ be respectively its first and second row, then we define $\pi(R)\in\mathbb C\mathbb P=\mathbb{P}^1$ by
\begin{equation*}
r_1=\pi(R)r_2,
\end{equation*}
with $\pi(R)=0$ if and only if $r_1=(0,0)$ and $\pi(R)=\infty$ if and only if $r_2=(0,0)$.

 Take any equivalence class $\textsf{C}=[C]\in \mathcal{M}_c(\lambda_0,a)$.
Let $1\leq k\leq 3$, then $|C(z)|$ has a simple zero at $z=x_k$, due to item \eqref{item:c3} in Definition \ref{def:connectionspace}, and thus $C(x_k)$, while nonzero, is not invertible. We define the coordinates
\begin{equation*}
\rho_k=\pi(C(x_k)),\quad (1\leq k\leq 3).
\end{equation*}
Note that $(\rho_1,\rho_2,\rho_3)$ are invariant under right multiplication by diagonal matrices and they are thus well-defined coordinates on $\mathcal{M}_c(\lambda_0,a)$. We denote the corresponding mapping by
\begin{equation}\label{eq:defrho}
    \rho: \mathcal{M}_c(\lambda_0,a)\rightarrow \mathbb{P}^1\times \mathbb P^1\times \mathbb P^1,\ \textrm{where}\ 
   \textsf{C}\overset{\rho}{\mapsto} (\rho_1,\rho_2,\rho_3).
\end{equation}

The coordinates $(\rho_1,\rho_2,\rho_3)$ are not independent. They represent points on a cubic surface defined below.
\begin{definition}\label{def:modulispace}
Define the cubic polynomial
\begin{align*}
T(p_1,p_2,p_3;\lambda_0,a)=&+\theta_q(+a_0,+a_1,+a_2)\left(\theta_q(\lambda_0)p_1p_2p_3-\theta_q(-\lambda_0)\right)\\
&-\theta_q(-a_0,+a_1,-a_2)\left(\theta_q(\lambda_0)p_1-\theta_q(-\lambda_0)p_2p_3\right)\\
&+\theta_q(+a_0,-a_1,-a_2)\left(\theta_q(\lambda_0)p_2-\theta_q(-\lambda_0)p_1p_3\right)\\
&-\theta_q(-a_0,-a_1,+a_2)\left(\theta_q(\lambda_0)p_3-\theta_q(-\lambda_0)p_1p_2\right),
\end{align*}
and its homogeneous form 
\begin{equation}\label{eq:thom}
    T_{hom}(p_1^x,p_1^y,p_2^x,p_2^y,p_3^x,p_3^y;\lambda_0,a)=p_1^yp_2^yp_3^yT\left(\frac{p_1^x}{p_1^y},\frac{p_2^x}{p_2^y},\frac{p_3^x}{p_3^y};\lambda_0,a\right).
  \end{equation}
  Then, using homogeneous coordinates $p_k=[p_k^x: p_k^y]\in \mathbb P^1$, for $1\le k\le 3$, the equation
  \begin{equation*}
      T_{hom}(p_1^x,p_1^y,p_2^x,p_2^y,p_3^x,p_3^y;\lambda_0,a)=0
  \end{equation*}
  defines a cubic surface in $\{(p_1,p_2,p_3)\in\mathbb{P}^1\times \mathbb P^1\times \mathbb P^1\}$, which we denote by $\mathcal{P}(\lambda_0,a)$.
\end{definition}

Our third main result is given by the following theorem.
\begin{theorem}\label{thm:mainresult3}
The range of the mapping $\rho$, defined in equation \eqref{eq:defrho}, is given by the cubic surface $\mathcal{P}(\lambda_0,a)$. On
$\mathcal{P}(\lambda_0,a)$, the mapping
\begin{equation*}
 \mathcal{M}_c(\lambda_0,a)\rightarrow \mathcal{P}(\lambda_0,a),\ \textrm{where}\ \textsf{C}\overset{\rho}{\mapsto} \rho(\textsf{C})
\end{equation*}
is a bijection and in particular the cubic surface $\mathcal{P}(\lambda_0,a)$ is the moduli space of $\mathcal{M}_c(\lambda_0,a)$ and thus of $q\Pfour^{\text{mod}}(\lambda_0,a)$ and $q\Pfour(\lambda_0,a)$.
\end{theorem}
\begin{remark}\label{rem:real}
It follows from Theorem \ref{thm:mainresult3} that $(\rho_1,\rho_2,\rho_3)\in \mathcal{P}(\lambda_0,a)$ parametrise the solution $f=(f^{\B{m}})_{m\in\mathbb{Z}}$ of $q\Pfour(\lambda_0,a)$. In the special case $a\in\mathbb{R}^3$, the transcendent $f$ is real-valued if and only if
$(\rho_1,\rho_2,\rho_3)\in \mathcal{P}(\lambda_0,a)\cap \mathbb P^1_{\mathbb R}\times \mathbb P^1_{\mathbb R}\times \mathbb P^1_{\mathbb R}$.
\end{remark}

\begin{remark}
  We note that in contrast to the monodromy manifolds associated with differential Painlev\'e equations \cite{chekhovmazzoccorubtsov2017}, the cubic surface $\mathcal{P}(\lambda_0,a)$ is non-affine. We also note that Chekhov \textit{et al.} \cite{chekhovmazzoccorubtsov2020} recently formulated a conjecture on the moduli spaces of $q$-difference equations corresponding to initial value spaces given by $A_0^{(1)*}$,  $A_1^{(1)}$, and $A_2^{(1)}$; these are stated as non-affine cubic surfaces in $\mathbb P^3$ (in the 2$^{\textrm{nd}}$, 3$^{\textrm{rd}}$ and 4$^{\textrm{th}}$ respective rows of \cite[Table 4]{chekhovmazzoccorubtsov2020}). Note that the conjecture does not include the case studied in the present paper, i.e., the initial value space $A_5^{(1)}$.

All discrete Painlev\'e equations arise as limits of Sakai's elliptic discrete Painlev\'e equation $q$P$(A_0^{(1)})$ and one may wish to know if our cubic surface $\mathcal{P}(\lambda_0,a)$  arises as a limit of those conjectured in \cite[Table 4]{chekhovmazzoccorubtsov2020}. There are two parts to the answer. One requires a change of coordinates between $\mathbb P^3$ and $\mathbb{P}^1\times \mathbb P^1\times \mathbb P^1$. The other requires the limits of parameters from $q$P$(A_0^{(1)*})$ to $q$P$(A_5^{(1)})$. However, the conjectured cubic surface for $q$P$(A_0^{(1)*})$ is given as the zero set of a cubic polynomial with arbitrary coefficients, which are not related explicitly to the parameters occurring in the $q$-discrete Painlev\'e equation. So we were unable to compare it to our surface, which is given explicitly in terms of such parameters.
\end{remark}

%% file: fuchsian.tex
\section{A class of Fuchsian systems}
\label{sec:class}
In this section, we study the direct and inverse monodromy problem concerning Equation \eqref{eq:laxspectral}. Consider
 linear systems of the form
\begin{equation}\label{eq:qdifferenceclass}
Y(qz)=A(z)Y(z),
\end{equation}
where $A(z)$ is a degree three $2\times 2$-matrix polynomial with the following properties.
        \begin{definition}[Properties of $A$]\label{pro:aa}
Let 	\begin{equation}\label{eq:matrixpolynomial}
	A(z)=A_0+zA_1+z^2A_2+z^3A_3,
	\end{equation}
        and assume
        \begin{enumerate}[label={{\rm (}a.\arabic*\rm{)}},ref=a.\arabic*]
	\item $A_0$ has eigenvalues $\{\pm i\}$,\label{item:p1}
	\item $A_3=qa_0^2a_2i\begin{pmatrix}
	\lambda & 0\\
	0 & \lambda^{-1}
	\end{pmatrix}$,\label{item:p2}
	\item $|A(z)|=(1-a_0z)(1+a_0z)(1-a_0a_2z)(1+a_0a_2z)(1-qz)(1+qz)$, \label{item:p3}
	\item $A(-z)=-\sigma_3A(z)\sigma_3$,\label{item:p4}		
\end{enumerate}
where $\lambda\in\mathbb{C}^*$ and $a=(a_0,a_1,a_2)\in\mathbb{C}^3$ are parameters such that $a_0a_1a_2=q$. Then $A$ is said to have {\em appropriate} properties. Furthermore, we define $\mathcal{F}(\lambda,a)$ to be the set of such appropriate matrices.
\end{definition}

 In Section \ref{subsection:direct}, we define fundamental solutions of Equation \eqref{eq:qdifferenceclass} near the origin and infinity and characterise the connection matrix relating them. In Section \ref{subsection:inverse}, we describe the inverse problem and define an equivalent Riemann-Hilbert problem.

\subsection{The Direct Monodromy Problem}
\label{subsection:direct}
Given $A\in \mathcal{F}(\lambda,a)$, note that Equation \eqref{eq:qdifferenceclass} is Fuchsian. We characterize the fundamental solutions of such an equation, and the corresponding connection matrix and monodromy mapping below. 

\subsubsection{Fundamental Solutions}
Carmichael \cite[Theorem 1]{carmichael1912} showed that Fuchsian $q$-difference systems have solutions with convergent expansions near $0$ and $\infty$. 

\begin{lemma}\label{lem:solzero} Let $\lambda\in\mathbb{C}^*$, $a\in\mathbb{C}^3$ with $a_0a_1a_2=q$, and $A(z)\in \mathcal{F}(\lambda,a)$. Define  $u=-A_{12}(0)$.
Then, for any $d\in\mathbb{C}^*$, we have
	\begin{equation}\label{eq:A0iv}
	A(0)=M_0 \begin{pmatrix}
	i & 0\\
	0 & -i
	\end{pmatrix}M_0^{-1}, \textrm{ where } M_0:=d\begin{pmatrix}
	u & 0\\
	0 & 1
	\end{pmatrix}\cdot\begin{pmatrix}
	i & -i\\
	1 & 1
	\end{pmatrix},
	\end{equation}	
and, there exists a unique $2\times 2$ matrix $\Phi_0(z)$, meromorphic on $\mathbb{C}^*$, such that
	\begin{align}\label{eq:Phizeroeqiv}
	\Phi_0(qz)&=A(z)\Phi_0(z)\begin{pmatrix}
	-i & 0\\
	0 & i
	\end{pmatrix},\\
	\Phi_0(z)&=M_0+\mathcal{O}\left(z\right),\  {\rm as}\ z\rightarrow 0,
	\end{align}
	with the following properties:
	\begin{enumerate}[label={{\rm (}z.\arabic*{\rm )}},ref=z.\arabic*]
		\item $\Phi_0(z)^{-1}$ is analytic on $\mathbb{C}$ and $\Phi_0(0)=M_0${\rm ;}
		\item $|\Phi_0(z)|=|M_0|\left(a_0z,-a_0z,a_0a_2z,-a_0a_2z,qz,-qz;q\right)_\infty^{-1}${\rm ;}
		\item\label{p:third0} $\Phi_0(-z)=-\sigma_3\Phi_0(z)\sigma_1$.
	\end{enumerate}
	In particular, it follows that
	\begin{equation*}
	Y_0(z)=\Phi_0(z)E_0(z),
	\end{equation*}
	defines a fundamental solution of \eqref{eq:qdifferenceclass}, for any meromorphic $2\times 2$ matrix function $E_0(z)$ on $\mathbb C^*$, satisfying
	\begin{equation*}
	E_0(qz)=\begin{pmatrix}
	i & 0\\
	0 & -i
	\end{pmatrix}E_0(z),\quad |E_0(z)|\not \equiv 0.
	\end{equation*}
\end{lemma}
\begin{proof}
Except for property \eqref{p:third0}, all the assertions of the lemma can be inferred directly from \cite{carmichael1912}*{Theorem 1}. To prove property \eqref{p:third0}, let
	\begin{equation*}
	\Phi(z):=-\sigma_3 \Phi_0(-z)\sigma_1.
	\end{equation*} 
	Then $\Phi(z)$ is analytic at $z=0$, with $\Phi(0)=M_0$, and straightforward calculation, using symmetry \eqref{item:p4}, shows that $\Phi(z)$ also satisfies Equation \eqref{eq:Phizeroeqiv}. By uniqueness, we must have $\Phi(z)=\Phi_0(z)$ and the lemma follows.
\end{proof}
Similarly, the following lemma provides a fundamental solution with a convergent expansion at infinity.

\begin{lemma}\label{lem:solinf} Let $\lambda\in\mathbb{C}^*$, $a\in\mathbb{C}^3$be given with $\lambda^2\notin q^\mathbb{Z}$ and $a_0a_1a_2=q$. For each $A(z)\in \mathcal{F}(\lambda,a)$, there exists a unique $2\times 2$ matrix  $\Phi_\infty(z)$, meromorphic on $\mathbb{C}^*$, such that
	\begin{align}\label{eq:Phiinfeqiv}
	\Phi_\infty(qz)&=\frac{1}{qa_0^2a_2i} z^{-3}A(z)\Phi_\infty(z)\begin{pmatrix}
	\lambda^{-1} & 0\\
	0 & \lambda
	\end{pmatrix},\\
	\Phi_\infty(z)&=I+\mathcal{O}\left(z^{-1}\right)\quad (z\rightarrow \infty),
	\end{align}
with the following properties:
	\begin{enumerate}[label={{\rm (}i.\arabic*{\rm )}},ref=i.\arabic*]
		\item $\Phi_\infty(z)$ is analytic on $\mathbb{P}^1\setminus\{0\}$ and $\Phi_\infty(\infty)=I$ {\rm ;}
		\item $|\Phi_\infty(z)|=\left(q/(a_0z),-q/(a_0z),q/(a_0a_2z),-q/(a_0a_2z),1/z,-1/z;q\right)_\infty${\rm ;}
		\item\label{p:third} $ \Phi_\infty(-z)=\sigma_3\Phi_\infty(z)\sigma_3$.
	\end{enumerate}
In particular,
\begin{equation*}
Y_\infty(z)=\Phi_\infty(z)E_\infty(z),
\end{equation*}
defines a fundamental solution of \eqref{eq:qdifferenceclass}, for any $2\times 2$ matrix $E_\infty(z)$, meromorphic on $\mathbb{C}^*$, satisfying	
\begin{equation*}
E_\infty(qz)=qa_0^2a_2iz^{-3}\begin{pmatrix}
\lambda & 0\\
0 & \lambda^{-1}
\end{pmatrix}E_\infty(z),\quad |E_\infty(z)|\not \equiv 0.
\end{equation*}
\end{lemma}
\begin{proof}
	Except for property \eqref{p:third}, all the contents of the lemma can be inferred directly from \cite{carmichael1912}*{Theorem 1}. To check the remaining property, let
	\begin{equation*}
	\Phi(z):=\sigma_3 \Phi_\infty(-z)\sigma_3.
	\end{equation*} 
	Then $\Phi(z)$ is analytic at $z=\infty$, with $\Phi(\infty)=I$, and straightforward calculation, using symmetry \eqref{item:p4}, shows that $\Phi(z)$ also solves equation \eqref{eq:Phiinfeqiv}. By uniqueness, we must have $\Phi(z)=\Phi_\infty(z)$ and the lemma follows.
\end{proof}
Note that $\Phi_\infty(z)$ is uniquely defined in Lemma \ref{lem:solinf}, whereas $\Phi_0(z)$ is only defined up to a constant multiplier in Lemma \ref{lem:solzero}.

\subsubsection{The Connection Matrix}
The fundamental solutions $Y_0(z)$ and $Y_\infty(z)$ constructed in Lemmas \ref{lem:solzero} and \ref{lem:solinf} are related by a connection matrix:
\begin{align*}
Y_\infty(z)&=Y_0(z)P(z),\\
P(z)&=Y_0(z)^{-1}Y_\infty(z),\\
&=E_0(z)^{-1}C(z)E_\infty(z),
\end{align*}
where
\begin{equation}\label{eq:assodefiiv}
C(z)=\Phi_0(z)^{-1}\Phi_\infty(z).
\end{equation}
Note that there is a great deal of freedom in choosing $E_0(z)$ and $E_\infty(z)$, which in turn implies that $P(z)$ is not rigidly defined. In contrast, the matrix $C(z)$ is rigidly defined up to a constant multiplier.

For this reason, we choose $C(z)$ to be the connection matrix associated with $A(z)$. This is in line with the Galoisian approach developed in \cite{putgalois1997}, where $E_0(z)$ and $E_\infty(z)$ are considered merely as formal scalings. The following result shows that such connection matrices $C$ are elements of the space $\mathfrak{C}(\lambda,a)$ defined in Definition \ref{def:connectionspace}.

\begin{lemma}\label{lem:connection}
Let $\lambda\in\mathbb{C}^*$ and $a\in\mathbb{C}^3$, be such that $\lambda^2\notin q^\mathbb{Z}$ and $a_0a_1a_2=q$. Suppose $A(z)\in \mathcal{F}(\lambda,a)$ is given, and let $\Phi_0(z)$ and $\Phi_\infty(z)$ be the corresponding solutions defined in Lemma \ref{lem:solzero} and \ref{lem:solinf} respectively. Then  the connection matrix given by $C(z)=\Phi_0(z)^{-1}\Phi_\infty(z)$ is an element of the space $\mathfrak{C}(\lambda,a)$.
\end{lemma}
\begin{proof}
Property \eqref{item:c2} follows from the $q$-difference equations which $\Phi_0(z)$ and $\Phi_\infty(z)$ satisfy. The remaining properties \eqref{item:c1}, \eqref{item:c3} and \eqref{item:c4} are direct consequences of the further properties of $\Phi_0(z)$ and $\Phi_\infty(z)$ given in 
Lemmas \ref{lem:solzero} and \ref{lem:solinf}.
\end{proof}

\subsubsection{The Monodromy Mapping} Recalling that the connection matrix is only uniquely defined up to scalar multiplication, we are led to the following definition of the space of equivalence classes of such matrices, under the condition of non-resonance.
\begin{definition}\label{def:monodromymap}
Let $\lambda\in\mathbb{C}^*$ and $a\in\mathbb{C}^3$ be given with $a_0a_1a_2=q$ and assume the non-resonant conditions \ref{eq:conditionnonresonant} are satisfied. Define the quotient space with respect to scalar multiplication by
\begin{equation*}
    M_c(\lambda,a)=\mathfrak{C}(\lambda,a)/\mathbb{C}^*,
\end{equation*}
and the monodromy mapping acting on $A(z)\in\mathcal{F}(\lambda,a)$ (see Definition \ref{pro:aa}) by
\begin{equation*}
M_{\mathcal{F}}:\mathcal{F}(\lambda,a)\rightarrow M_c(\lambda,a),\ \textrm{where}\ A(z)\overset{M_{\mathcal{F}}}{\mapsto} C(z),
\end{equation*}
where the image value $C(z)$ is the connection matrix corresponding to the Fuchsian system \ref{eq:qdifferenceclass} via Lemma \ref{lem:connection}.
\end{definition}

The condition of non-resonance  \eqref{eq:conditionnonresonant} in the above definition is required for the proof of the following proposition.
\begin{proposition}\label{pro:injective}
The monodromy mapping $M_{\mathcal{F}}$ defined in \ref{def:monodromymap} is injective.
\end{proposition}
\begin{proof}
Let $A(z),\widetilde{A}(z)\in \mathcal{F}(\lambda,a)$, and denote corresponding $\Phi_0(z)$, $\Phi_\infty(z),C(z)$ and $\widetilde{\Phi}_0(z)$, $\widetilde{\Phi}_\infty(z),\widetilde{C}(z)$ as defined in Lemmas \ref{lem:solzero}, \ref{lem:solinf} and \ref{lem:connection} respectively.

Suppose $M_{\mathcal{F}}(\widetilde{A})=M_{\mathcal{F}}(A)$, then there exists a $c\in\mathbb{C}^*$ such that $\widetilde{C}(z)=c\,C(z)$, hence
\begin{align}
G(z):&=\widetilde{\Phi}_\infty(z)\Phi_\infty(z)^{-1}\label{eq:ginf}\\
&=\widetilde{\Phi}_0(z)\widetilde{C}(z)C(z)^{-1}\Phi_0(z)^{-1}\nonumber\\
&=c\,\widetilde{\Phi}_0(z)\Phi_0(z)^{-1}.\label{eq:g0}
\end{align}
From equation \eqref{eq:ginf} and Lemma \ref{lem:solinf}, it is clear that $G(z)$ is analytic on
\begin{equation*}
\mathbb{C}^*\setminus\left(q^{\mathbb{N}^*}\{x_1,\ldots,x_{6}\}\right),
\end{equation*}
where $x_1,\ldots,x_6$ are as defined in \eqref{eq:x1x6}. Similarly, from \eqref{eq:g0} and Lemma \ref{lem:solzero}, it follows that $G(z)$ is analytic on
\begin{equation*}
\mathbb{C}\setminus\left(q^{-\mathbb{N}}\{x_1,\ldots,x_{6}\}\right).
\end{equation*}
We conclude that $G(z)$ is analytic on the complement of
\begin{equation}\label{eq:complement}
\big(q^{\mathbb{N}^*}\{x_1,\ldots,x_{6}\}\big)\cap \big(q^{-\mathbb{N}}\{x_1,\ldots,x_{6}\}\big),
\end{equation}
in $\mathbb{C}$. However, precisely because of the non-resonant conditions \eqref{eq:conditionnonresonant}, the intersection in \eqref{eq:complement} is empty, so $G(z)$ is analytic on $\mathbb{C}$. Finally from equation \eqref{eq:ginf} and Lemma \ref{lem:solinf}, it follows that $G(z)=I+\mathcal{O}(z^{-1})$ as $z\rightarrow \infty$, and hence $G(z)\equiv I$ by Liouville's theorem. Therefore $\widetilde{\Phi}_\infty(z)=\Phi_\infty(z)$, and hence $\widetilde{A}(z)=A(z)$, by equation \eqref{eq:Phiinfeqiv}, which gives the desired result.
\end{proof}

\subsection{The Inverse Monodromy Problem}\label{subsection:inverse}
In this section, we consider the surjectivity of the monodromy mapping, which is a more delicate issue than its injectivity (as it is the content of the $q$-analog of Hilbert's 21st problem). Birkhoff \cite{birkhoffgeneralized1913} gave a comprehensive description of this problem in the generic non-resonant case.

Considering the class of Fuchsian systems \eqref{eq:qdifferenceclass}, we formulate the main inverse problem as follows.
\begin{prob}[The Inverse Monodromy Problem]\label{prob:inverse}
Let $\lambda\in\mathbb{C}^*$ and $a\in\mathbb{C}^3$, satisfying $a_0a_1a_2=q$, such that the non-resonant conditions \eqref{eq:conditionnonresonant} are satisfied. Given a monodromy datum $\textsf{C}=[C(z)]\in M_c(\lambda,a)$, construct a Fuchsian system \eqref{eq:qdifferenceclass}, whose associated connection matrix (modulo a constant) is $\textsf{C}$.
\end{prob}
In Proposition \ref{pro:imrhequivalence}, we show that this inverse problem is equivalent to Riemann-Hilbert problem $\text{RHP}^{\B{0}}(\gamma,C)$, defined in Definition \ref{def:RHP1and2}, for any admissible curve $\gamma$. But first we prove that Riemann-Hilbert problem $\text{RHP}^{\B{m}}(\gamma,C)$ has at most one solution, for any $m\in\mathbb{Z}$.

\begin{lemma}\label{lem:rhuniqueness}
	For any $m\in\mathbb{Z}$, if $\textnormal{RHP}^{\B{m}}(\gamma,C)$ has a solution $Y^{\B{m}}(z)$, then it is unique. Furthermore, let $c\in\mathbb{C}^*$ be defined by \eqref{item:c3}, then the determinant $\Delta(z)=|Y^{\B{m}}(z)|$ equals
	\begin{equation}\label{eq:psidet}
	\Delta(z)=\begin{cases}
	\left(qx_1/z,\ldots,qx_{6}/z;q\right)_\infty & \text{if $z\in D_+$,}\\
	c^{-1}\left(z/x_1,\ldots z/x_{6};q\right)_\infty^{-1} & \text{if $z\in D_-$.}
	\end{cases}
	\end{equation}
	In particular $Y^{\B{m}}(z)$ is globally invertible on $\mathbb{C}\setminus \gamma$.
\end{lemma}
\begin{proof}
Note that the determinant $\Delta(z)=|Y^{\B{m}}(z)|$ solves the following scalar Riemann-Hilbert problem:
\begin{enumerate}[label={{\rm (s.\arabic*)}},ref=s.\arabic*]
\item  $\Delta(z)$ is analytic on $\mathbb{C}\setminus\gamma$.
		\item  $\Delta(z)$ has continuous boundary values $\Delta_-(z)$ and $\Delta_+(z)$ for $z\in \gamma$,  related by the jump condition
		\begin{equation*}
		 \Delta_+(z)=\Delta_-(z)c\, \theta_q(z/x_1,\ldots,z/x_{6})\quad (z\in \gamma).
		 \end{equation*}
		\item  $\Delta(z)$ has the following asymptotic behaviour near infinity,
		\begin{equation*}\Delta(z)=1+\mathcal{O}\left(z^{-1}\right)\quad (z\rightarrow \infty).
		\end{equation*}
	\end{enumerate}
	We now show that this problem has a unique solution, given by \eqref{eq:psidet}. 
First note that $|Y^{\B{m}}(z)|$ vanishes nowhere and hence $Y^{\B{m}}(z)^{-1}$ is an analytic function on $\mathbb{C}\setminus\gamma$.
	Therefore, given another solution $\Psi(z)$ of Riemann-Hilbert problem $\textnormal{RHP}^{\B{m}}(\gamma,C)$, then
	\begin{equation*}
	H(z):=\Psi(z)Y^{\B{m}}(z)^{-1}
	\end{equation*}
	defines an analytic function on $\mathbb{C}\setminus\gamma$. Since $Y^{\B{m}}(z)$ and $\Psi(z)$ satisfy the same jump condition on $\gamma$,   $H(z)$ has analytic continuation to $\mathbb{C}$. Furthermore the asymptotic behaviour near infinity of both solutions implies $H(z)=I+\mathcal{O}\left(z^{-1}\right)$ as $z\rightarrow \infty$, and we conclude $H(z)\equiv I$ by Liouville's theorem.
	So $Y^{\B{m}}(z)=\Psi(z)$ and uniqueness follows. 
\end{proof}

\begin{proposition} \label{pro:imrhequivalence}
Let $\lambda\in\mathbb{C}^*$ and $a\in\mathbb{C}^3$, satisfying $a_0a_1a_2=q$, such that the non-resonant conditions \eqref{eq:conditionnonresonant} are satisfied. Given a monodromy datum $\textsf{C}=[C(z)]\in M_c(\lambda,a)$,
the inverse monodromy problem \ref{prob:inverse} is equivalent to  $\text{RHP}^{\B{0}}(\gamma,C)$, for any admissible curve $\gamma$, in the following sense.
	\begin{enumerate}[label={{\rm (\roman *)}},leftmargin=*]
		\item If $A(z)\in\mathcal{F}(\lambda,a)$ is a solution of the inverse monodromy problem \ref{prob:inverse}, then there exists a unique value of $d\in\mathbb{C}^*$ in Lemma \ref{lem:solzero} for which the corresponding matrix function $\Phi_0^{\B{m}}(z)$, together with the matrix function $\Phi_\infty^{\B{m}}(z)$ constructed in Lemma \ref{lem:solinf}, define a solution
		\begin{equation}\label{eq:psidef}
		\Psi(z):=\begin{cases}
		\Phi_\infty(z) &\text{if $z\in D_+$,}\\
	\Phi_0(z) &\text{if $z\in D_-$,}
		\end{cases}
		\end{equation}
		of Riemann-Hilbert problem $\text{RHP}^{\B{0}}(\gamma,C)$.
		\item Conversely, suppose $\Psi(z)$ is a solution of Riemann-Hilbert problem\hfill \\ $\text{RHP}^{\B{0}}(\gamma,C)$. Then  defining the solutions
		\begin{equation}\label{eq:phiinfzerodef}
		\Psi_\infty(z):=\begin{cases}
		\Psi(z) &\text{if $z\in D_+$,}\\
		\Psi(z) C(z) &\text{if $z\in D_-$,}
		\end{cases}\quad \Psi_0(z):=\begin{cases}
		\Psi(z) C(z)^{-1} &\text{if $z\in D_+$,}\\
		\Psi(z) &\text{if $z\in D_-$,}
		\end{cases}
		\end{equation}
		then $\Psi_\infty(z)$ and $\Psi_0(z)^{-1}$ are related by
		\begin{equation}\label{eq:psiinfpsi0}
		\Psi_\infty(z)=\Psi_0(z)C(z),
		\end{equation}
		and
		\begin{align}
		A(z):&=qa_0^2a_2iz^3\Psi_\infty(qz)\begin{pmatrix}
		\lambda & 0\\
		0 & \lambda^{-1}
		\end{pmatrix}\Psi_\infty(z)^{-1}\label{eq:defainf}\\
		&=\Psi_0(qz)\begin{pmatrix}
		i & 0\\
		0 & -i
		\end{pmatrix}\Psi_0(z)^{-1},\label{eq:defa0}
		\end{align}
		defines a solution $A(z)\in\mathcal{F}(\lambda,a)$ of the inverse monodromy problem \ref{prob:inverse}.
	\end{enumerate}
\end{proposition}
\begin{proof}
Consider part (i). Take $d\in\mathbb{C}^*$ and define $\Phi_0(z)$ and $\Phi_\infty(z)$ as in Lemma \ref{lem:solzero} and \ref{lem:solinf}, such that $\Phi_\infty(z)=\Phi_0(z)C(z)$; see equation \eqref{eq:assodefiiv}.
Then equation \eqref{eq:psidef} defines a matrix function $\Psi(z)$, which satisfies the jump condition of $\text{RHP}^{\B{0}}(\gamma,C)$. Furthermore, by Lemmas \ref{lem:solzero} and \ref{lem:solinf} it also satisfies the analyticity and asymptotic conditions of $\text{RHP}^{\B{0}}(\gamma,C)$ and thus solves it.

Now consider part (ii). Clearly $\Psi_\infty(z)$ and $\Psi_0(z)^{-1}$, defined by equations \eqref{eq:phiinfzerodef}, are analytic on $\mathbb{C}^*$ and $\mathbb{C}$ respectively, and related by \eqref{eq:psiinfpsi0}. Hence, defining $A(z)$ by \eqref{eq:defainf}, equation \eqref{eq:defa0} follows from property \eqref{item:c2} of $C(z)$.

We proceed to show that $A(z)\in \mathcal{F}(\lambda,a)$.
 From equations \eqref{eq:defainf} and \eqref{eq:defa0} we obtain respectively that $A(z)$ is analytic on $q^{-1}D_+$ and $D_-$. Furthermore, note that 
\begin{equation*}
A(z)=qa_0^2a_2iz^3\Psi_0(qz)C(qz)\lambda^{\sigma_3}\Psi_\infty(z)^{-1},
\end{equation*}
from which it follows that $A(z)$ is also analytic on $\mathbb{C}\setminus(q^{-1}D_+\cup D_-)$. Hence $A(z)$ is analytic on $\mathbb{C}$. It follows from the asymptotic behaviour of $\Psi(z)$ that
\begin{equation*}
A(z)=qa_0^2a_2iz^3\left(\lambda^{\sigma_3}+\mathcal{O}\left(z^{-1}\right)\right)\quad \ \textrm{as}\ z\rightarrow \infty,
\end{equation*}
and thus $A(z)$ is a matrix polynomial of degree three, satisfying property \eqref{item:p2}.

Due to equation \eqref{eq:defa0}, we know that $A(0)$ has eigenvalues $\{\pm i\}$, so $A(z)$ satisfies property \eqref{item:p1} of Definition \ref{pro:aa}.
Furthermore, using the explicit expression for the determinant $|\Psi(z)|$ in Lemma \ref{lem:rhuniqueness}, property \eqref{item:p3} of Definition \ref{pro:aa} easily follows.
We proceed with deriving the remaining property \eqref{item:p4} of Definition \ref{pro:aa}.
Recall that $C(z)$ has the symmetry $C(-z)=-\sigma_1C(z)\sigma_3$ and $\gamma$ is reflection-invariant,
hence
\begin{equation*}
\widetilde{\Psi}(z)=\begin{cases}
\sigma_3\Psi(-z)\sigma_3 &\text{if $z\in D_+$,}\\
-\sigma_3\Psi(-z)\sigma_1,&\text{if $z\in D_-$,}
\end{cases}
\end{equation*}
also defines a solution of $\text{RHP}^{\B{0}}(\gamma,C)$. By the uniqueness result in Lemma \ref{lem:rhuniqueness}, we must have $\widetilde{\Psi}(z)=\Psi(z)$ and thus
\begin{equation}\label{eq:psiinfsym}
\Psi_\infty(-z)=\sigma_3\Psi_\infty(z)\sigma_3,\quad \Psi_0(-z)=-\sigma_3\Psi_0(z)\sigma_1,
\end{equation} 
giving $A(-z)=-\sigma_3A(z)\sigma_3$,
which is precisely \eqref{item:p3} of Definition \ref{pro:aa}. 
We conclude that $A(z)\in\mathcal{F}(\lambda,a)$. Furthermore, note that $\Psi_0(z)=\Phi_0(z)$, for a unique choice of $d\in\mathbb{C}^*$ in Lemma \ref{lem:solzero}, and $\Psi_\infty(z)=\Phi_\infty(z)$. Therefore
 $A(z)$ defines a solution of the inverse monodromy problem \ref{prob:inverse}.
\end{proof}

%% file: isomonodromic.tex
\section{Isomonodromic deformation}\label{sec:isomonodromic}
In this section, we consider isomonodromic deformation of Fuchsian systems of the form \eqref{eq:qdifferenceclass} as $\lambda\rightarrow q\lambda$ and prove Theorems \ref{thm:mainresult1} and \ref{thm:rhmap}. Clearly the space $\mathfrak{C}(\lambda,a)$, defined in Definition \ref{def:connectionspace}, is not invariant under $\lambda\rightarrow q\lambda$. However, we have the following bijective mapping
\begin{equation*}
  \tau:
    \mathfrak{C}(\lambda,a)\rightarrow \mathfrak{C}(q\lambda,a),\ \textrm{where}\  C(z)\overset{\tau}{\mapsto} \sigma_3C(z)z^{-\sigma_3}.
 \end{equation*}
Note that this mapping commutes with scalar multiplication and right-multiplication by invertible diagonal matrices. It therefore induces bijective mappings on $M_c(\lambda,a)$ and $\mathcal{M}_c(\lambda,a)$, see Definitions \ref{def:monodromymap} and \ref{def:monodromysurface} respectively, which we continue to refer to as $\tau$, i.e.,
 \begin{equation}\label{eq:taudef}
   \tau:    M_c(\lambda,a)\rightarrow M_c(q\lambda,a),\quad  \tau:\mathcal{M}_c(\lambda,a)\rightarrow \mathcal{M}_c(q\lambda,a).
  \end{equation}

We call a deformation $\lambda\rightarrow q\lambda$ of the Fuchsian system \eqref{eq:qdifferenceclass} isomonodromic if its action trivially deforms the monodromy or connection matrix by the action of $\tau$. 
Correspondingly, we define the following inverse problem.
\begin{prob}[Generalised Inverse Monodromy Problem]\label{prob:inversem}
	Let $\lambda_0\in\mathbb{C}^*$ and $a\in\mathbb{C}^3$, satisfying $a_0a_1a_2=q$, be such that the non-resonant conditions \eqref{eq:conditionnonresonant} are satisfied. Given a monodromy datum $\textsf{C}=[C(z)]\in M_c(\lambda_0,a)$ and an $m\in\mathbb{Z}$, construct a Fuchsian system 
	\begin{equation*}
	Y(qz)=A^{\B{m}}(z)Y(z),\quad A^{\B{m}}(z)\in \mathcal{F}(q^m\lambda_0,a),
	\end{equation*}
	 whose monodromy equals the $m$-fold composition $\tau^m(M)$.
\end{prob}
In this section we prove that the aforementioned isomonodromic deformation as $\lambda\rightarrow q\lambda$ is equivalent to the $q\Pfour^{\text{mod}}$ time-evolution. In particular we show that solving the inverse problem \ref{prob:inversem} for a given monodromy datum is equivalent to computing the values of a particular $q\Pfour^{\text{mod}}(\lambda_0,a)$ transcendent.
 
Associated with the inverse problem \ref{prob:inversem}, is the Riemann-Hilbert Problem defined in Definition \ref{def:RHP1and2}. In Section \ref{subsec:solve} we show that the inverse problem \ref{prob:inversem} is equivalent to this RHP. Furthermore we show that the RHP is always solvable, for at least one value of $m\in\mathbb{Z}$. In Section \ref{subsection:timedef} we compute the isomonodromic deformation of the class of Fuchsian systems \eqref{eq:qdifferenceclass} explicitly using the associated RHP, yielding Theorem \ref{thm:mainresult1} in Section \ref{s:mr}. Finally, in Section \ref{section:rhcorrespondence}, we prove Theorem \ref{thm:rhmap}.

\subsection{Solvability of the Generalised Inverse Monodromy Problem}\label{subsec:solve}
In this Section, we prove the solvability of inverse problem \ref{prob:inversem} for at least one value of $m\in\mathbb{Z}$. Firstly, in the following proposition, we make the equivalence of the generalised inverse monodromy problem \ref{prob:inversem} and the RHP defined in Definition \ref{def:RHP1and2} explicit.
\begin{proposition} \label{cor:imrhequivalencem}
	Let $\lambda_0\in\mathbb{C}^*$ and $a\in\mathbb{C}^3$, satisfying $a_0a_1a_2=q$, such that the non-resonant conditions \eqref{eq:conditionnonresonant} are satisfied. Given a monodromy datum $M=[C(z)]\in M_c(\lambda,a)$ and an $m\in\mathbb{Z}$,
	then, for any choice of admissible curve, the generalised inverse monodromy problem \ref{prob:inversem} is equivalent to Riemann-Hilbert problem $\textnormal{RHP}^{\B{m}}(\gamma,C)$, in the following sense.
	\begin{enumerate}
		\item[\rm(i\rm)] If $A^{\B{m}}(z)\in\mathcal{F}(q^m\lambda_0,a)$ is a solution of the inverse monodromy problem \ref{prob:inversem}, then there exists a unique value of $d=d_m\in\mathbb{C}^*$ in Lemma \ref{lem:solzero} for which the corresponding matrix function $\Phi_0^{\B{m}}(z)$, together with the matrix function $\Phi_\infty^{\B{m}}(z)$ constructed in Lemma \ref{lem:solinf}, define a solution
		\begin{equation}\label{eq:psidefm}
		Y^{\B{m}}(z):=\begin{cases}
		\Phi_\infty^{\B{m}}(z)z^{m\sigma_3} &\text{if $z\in D_+$,}\\
		\Phi_0^{\B{m}}(z) &\text{if $z\in D_-$,}
		\end{cases}
		\end{equation}
	    of $\textnormal{RHP}^{\B{m}}(\gamma,C)$.
		\item[\rm(ii\rm)] Conversely, suppose $Y^{\B{m}}(z)$ is a solution of $\textnormal{RHP}^{\B{m}}(\gamma,C)$, writing
		\begin{align*}
		&\Psi_\infty^{\B{m}}(z):=\begin{cases}
		Y^{\B{m}}(z)z^{-m\sigma_3} &\text{if $z\in D_+$,}\\
		Y^{\B{m}}(z) C(z)z^{-m\sigma_3} &\text{if $z\in D_-$,}
              \end{cases}\\
                 & \Psi_0^{\B{m}}(z):=\begin{cases}
		Y^{\B{m}}(z) \sigma_3^{-m}C(z)^{-1} &\text{if $z\in D_+$,}\\
		Y^{\B{m}}(z)\sigma_3^{-m} &\text{if $z\in D_-$,}
		\end{cases}
		\end{align*}
		then $\Psi_\infty^{\B{m}}(z)$ and $\Psi_0^{\B{m}}(z)^{-1}$ are related by
		\begin{equation}\label{eq:psiinfpsi0m}
		\Psi_\infty^{\B{m}}(z)=\Psi_0^{\B{m}}(z)\sigma_3^{m}C(z)z^{-m\sigma_3},
		\end{equation}
		and, denoting
		\begin{equation}\label{eq:notation}
		\kappa=qa_0^2a_2i,\quad \lambda_m=q^m\lambda_0,
		\end{equation}
		the matrix polynomial
		\begin{align}
		A^{\B{m}}(z):&=z^3\Psi_\infty^{\B{m}}(qz)\begin{pmatrix}
		\kappa \lambda_m & 0\\
		0 & \kappa \lambda_m^{-1}
		\end{pmatrix}\Psi_\infty^{\B{m}}(z)^{-1}\label{eq:defainfm}\\
		&=\Psi_0^{\B{m}}(qz)\begin{pmatrix}
		i & 0\\
		0 & -i
		\end{pmatrix}\Psi_0^{\B{m}}(z)^{-1},\label{eq:defa0m}
		\end{align}
		defines a solution $A^{\B{m}}(z)\in\mathcal{F}(q^m\lambda_0,a)$ of the inverse monodromy problem \ref{prob:inversem}.
	\end{enumerate}
\end{proposition}
\begin{proof}
Fix any $m\in\mathbb{Z}$, then $Y^{\B{m}}(z)$ is a solution of $\textnormal{RHP}^{\B{m}}(\gamma,C)$ if and only if $\Psi(z)=Y^{\B{m}}(z) S(z)^{-1}$ is a solution of $\textnormal{RHP}^{\B{0}}(\gamma,\tau^m(C))$, where
\begin{equation*}
S(z)=\begin{cases}
z^{m\sigma_3} & \text{if $z\in D_+$,}\\
\sigma_3^m & \text{if $z\in D_-$.}
\end{cases}
\end{equation*}
Therefore, statements (i) and (ii) are equivalent to the respective numbered statements in Proposition
\ref{pro:imrhequivalence}, after the substitutions $C(z)\mapsto \tau^m(C(z))$, $\lambda\mapsto q^m\lambda_0$ and $\Psi(z)\mapsto Y^{\B{m}}(z) S(z)^{-1}$ are made.
The proposition is thus a direct corollary of Proposition \ref{pro:imrhequivalence}.
\end{proof}

 In the remainder of this section, we give a classical argument, going back to Birkhoff \cite{birkhoffgeneralized1913}, showing that $\textnormal{RHP}^{\B{m}}(\gamma,C)$ has a solution $Y^{\B{m}}(z)$, for at least one value of $m\in\mathbb{Z}$. Firstly, we state a special case of the ``Preliminary Theorem" in Birkhoff \cite{birkhoffgeneralized1913}.

\begin{lemma}\label{lem:riemannhilbertbirkhoff}
	Let $\gamma$ be an oriented analytic Jordan curve in $\mathbb{P}^1$ and let $D_-$ and $D_+$ denote the interior and exterior of  $\gamma$ in $\mathbb{P}^1$ respectively. Let $C(z)$ be a $2\times 2$ matrix function, analytic on $\gamma$, such that $|C(z)|\neq 0$ on $\gamma$. Then, for any $\alpha\in\mathbb{P}^1\setminus \gamma$, there exists a $2\times 2$ matrix function $Y(z)$, satisfying
	\begin{itemize}
	\item $Y(z)$ is analytic on $\mathbb{P}^1\setminus (\gamma\cup\{\alpha\})$ and meromorphic at $z=\alpha$.
		\item $Y(z)$ has continuous boundary values $Y_+(z)$ and $Y_-(z)$ as $z$ approaches $\gamma$ from $D_+$ and $D_-$ respectively, which are related by the jump condition
		\begin{equation}\label{eq:jumplem}
		Y_+(z)=Y_-(z)C(z).
		\end{equation}
		\item The determinant $|Y(z)|$ does not vanish on $\mathbb{P}^1\setminus (\gamma\cup\{\alpha\})$, and neither $|Y_-(z)|$ nor $|Y_+(z)|$ vanishes on $\gamma$.
	\end{itemize}
\end{lemma}
\begin{proof}
	This is a special case of the ``Preliminary Theorem" in Birkhoff \cite{birkhoffgeneralized1913}.
\end{proof}

Next we state and prove a special case of the main conclusion of paragraph 21 in Birkhoff \cite{birkhoffgeneralized1913}.
\begin{lemma}\label{lem:existence}
$\textnormal{RHP}^{\B{m}}(\gamma,C)$ has a solution $Y^{\B{m}}(z)$, for at least one value of $m\in\mathbb{Z}$.
\end{lemma}
\begin{proof} For the convenience of the reader, we paraphrase Birkhoff's proof of the conclusion in paragraph 21 of \cite{birkhoffgeneralized1913}  for our special case.

	Firstly, note that the matrix $C(z)$ is analytic and $|C(z)|$ does not vanish on $\gamma$. We may thus apply Lemma \ref{lem:riemannhilbertbirkhoff} with $\alpha=\infty$, which gives a matrix function $Y(z)$ that satisfies the analyticity and jump condition in $\textnormal{RHP}^{\B{m}}(\gamma,C)$, $m\in\mathbb{Z}$. It remains to normalise $Y(z)$ appropriately.
	
	Firstly, we compare the determinant $|Y(z)|$ with $\Delta(z)$, defined in equation \eqref{eq:psidet}.
	Note that $d(z)=|Y(z)|/\Delta(z)$ defines a non-vanishing analytic function on $\mathbb{C}$ which is meromorphic at $z=\infty$, so $d(z)\equiv d_0\in\mathbb{C}^*$ is constant. In particular
	\begin{equation}\label{eq:det}
	|Y(z)|=d_0\left(qx_1/z,\ldots,qx_{6}/z;q\right)_\infty\quad (z\in D_+).
	\end{equation}
	
	For any matrix function $H(z)$, replacing $Y(z)\mapsto \widetilde{Y}(z)=H(z)Y(z)$ in the above, all analytic properties in Lemma \ref{lem:riemannhilbertbirkhoff} are conserved if and only if
	$H(z)$ is a matrix polynomial with $|H(z)|\equiv h\in\mathbb{C}^*$ constant.  It therefore suffices to find an appropriate such $H(z)$, so that
	\begin{equation}\label{eq:psiplusformdesire}
	\widetilde{Y}(z)=\left(U+\mathcal{O}(z^{-1})\right)\begin{pmatrix}
	z^m & 0\\
	0 & z^{-m}
	\end{pmatrix}\quad (z\rightarrow \infty),
	\end{equation}
	for an $m\in\mathbb{Z}$ and $U\in GL_2(\mathbb{C})$.	Indeed $U^{-1}\widetilde{Y}(z)$ then defines a solution to Riemann-Hilbert Problem $\textnormal{RHP}^{\B{m}}(\gamma,C)$.
	
	To this end, determine the unique $m_1,m_2\in\mathbb{Z}$ and $U\in\mathbb{C}^{2\times 2}$, with both columns nonzero, such that
	\begin{equation}\label{eq:psiplusform}
    Y(z)=\left(U+\mathcal{O}(z^{-1})\right)\begin{pmatrix}
	z^{m_1} & 0\\
	0 & z^{m_2}
	\end{pmatrix}\quad (z\rightarrow \infty).
	\end{equation}
	By equation \eqref{eq:det}, we must have $K(Y):=m_1+m_2\geq 0$. Furthermore, note that, again by \eqref{eq:det}, $K=0$ if and only if $U$ is invertible, in which case we are done as $Y(z)$ has the desired form \eqref{eq:psiplusformdesire}. We proceed to show that, if $K(Y)>0$, then there exists a matrix polynomial $G(z)$ with $|G(z)|\equiv 1$, such that $Y'(z):= G(z)Y(z)$ will have a strictly smaller $K$ value then $Y(z)$, i.e. $K(Y')<K(Y)$.
	
	So assume $K(Y)>0$. Then $U$ is not invertible and has two nonzero columns, hence there exists an $M\in GL_2(\mathbb{C})$ such that
	\begin{equation*}
	MU=\begin{pmatrix}
	0 & 0\\
	u_{21}' & u_{22}'
	\end{pmatrix},
	\end{equation*}
	for some $u_{21}',u_{22}'\in\mathbb{C}^*$. Let $l>0$ be such that
	\begin{equation*}
	MY(z)=\begin{pmatrix}
	z^{-l} & 0\\
	0 & 1
	\end{pmatrix}\left(\begin{pmatrix}
	u_{11}' & u_{12}'\\
	u_{21}' & u_{22}'
	\end{pmatrix}+\mathcal{O}\left(z^{-1}\right)\right)\begin{pmatrix}
	z^{m_1} & 0\\
	0 & z^{m_2}
	\end{pmatrix}\quad (z\rightarrow \infty),
	\end{equation*}
	for some $u_{11}',u_{12}'\in\mathbb{C}$ not both equal to zero. Next we multiply from the left by a matrix polynomial
	\begin{equation*}
	R(z)=\begin{pmatrix}
	1 & 0\\
	rz^l & 1
	\end{pmatrix},
	\end{equation*}
	which gives
	\begin{equation*}
	R(z)MY(z)=\begin{pmatrix}
	z^{-l} & 0\\
	0 & 1
	\end{pmatrix}\left(\begin{pmatrix}
	u_{11}' &  u_{12}'\\
	u_{21}'+ru_{11}' & u_{22}'+ru_{12}'
	\end{pmatrix}+\mathcal{O}\left(z^{-1}\right)\right)\begin{pmatrix}
	z^{m_1} & 0\\
	0 & z^{m_2}
	\end{pmatrix},
	\end{equation*}
	as $z\rightarrow \infty$. Determine $i\in\{1,2\}$ such that $u_{1i}'\neq 0$ and choose $r\in\mathbb{C}$ such that $u_{2i}'+ru_{1i}'=0$. Set $G(z)=M^{-1}R(z)M$, then $G(z)$ is a matrix polynomial satisfying $|G(z)|\equiv 1$. Let $Y'(z):= G(z)Y(z)$, then $K(Y')<K(Y)$.
	
	Applying the above argument recursively, we obtain matrix polynomials $G_0(z)$, $G_1(z)$,$\ldots$, $G_k(z)$, each with unit determinant, for some $k\in\mathbb{N}$, such that, setting
	\begin{equation*}
	Y_0(z):=Y(z),\quad Y_{s+1}=G_s(z)Y_s(z)\quad (0\leq s\leq k-1),
	\end{equation*}
	we have $K(Y_0)>K(Y_1)>\ldots>K(Y_k)=0$. Define
	\begin{equation*}
	\widetilde{Y}(z)=H(z)Y(z),\quad H(z)=G_{k-1}(z)G_{k-2}(z)\cdot\ldots\cdot G_{0}(z),
	\end{equation*}
	then equation \eqref{eq:psiplusformdesire} holds true, yielding the lemma.
\end{proof}

\subsection{Isomonodromic Deformation and the Riemann-Hilbert Problem}\label{subsection:timedef}
In this section we prove Theorem \ref{thm:mainresult1}. Firstly, we relate the coefficients of $A^{\B{m}}(z)$ explicitly to the coefficients in the expansion around infinity of $Y^{\B{m}}(z)$.

Let $m\in\mathbb{Z}$ and suppose the solution $Y^{\B{m}}(z)$ of the Riemann-Hilbert problem $\textnormal{RHP}^{\B{m}}(\gamma,C)$ exists.  We know that there exist a unique $g^{\B{m}}\in \mathcal{G}(a)$ and $u_m\in\mathbb{C}^*$, such that $A^{\B{m}}(z)\in \mathcal{F}(q^m\lambda_0,a)$ is given by
\begin{align}\label{eq:Ameq}
A^{\B{m}}(z)=&\mathcal{A}(z;\lambda_m,g^{\B{m}},u_m)\\
=&\begin{pmatrix}
0 & -u_m\\u_m^{-1} & 0
\end{pmatrix}+
z\begin{pmatrix}
ig_1^{\B{m}}\lambda_m & 0\\0 & ig_2^{\B{m}}\lambda_m^{-1}
\end{pmatrix}\nonumber\\
&+z^2\begin{pmatrix}
0 & -u_mg_3^{\B{m}}\\u_m^{-1}g_4^{\B{m}} & 0
\end{pmatrix}+z^3\begin{pmatrix}
\kappa \lambda_m & 0\\
0 & \kappa \lambda_m^{-1}
\end{pmatrix},\nonumber
\end{align}
where we again used the notation \eqref{eq:notation}.
Also there exist unique matrices $U^{\B{m}}$, $V^{\B{m}}$, $W^{\B{m}}$ $\in\mathbb{C}^{2\times 2}$ such that
\begin{equation}\label{eq:widehatpsi}
Y^{\B{m}}(z)=\left(I+z^{-1}U^{\B{m}}+z^{-2}V^{\B{m}}+z^{-3}W^{\B{m}}+\mathcal{O}\left(z^{-4}\right)\right)z^{m\sigma_3},
\end{equation}
as $z\rightarrow \infty$ and thus
\begin{equation*}
\Psi_\infty^{\B{m}}(z)=I+z^{-1}U^{\B{m}}+z^{-2}V^{\B{m}}+z^{-3}W^{\B{m}}+\mathcal{O}\left(z^{-4}\right)\quad (z\rightarrow \infty).
\end{equation*}
Due to Equation \eqref{eq:psiinfsym}, we must have
\begin{equation*}
\sigma_3U^{\B{m}}\sigma_3=-U^{\B{m}},\quad \sigma_3V^{\B{m}}\sigma_3=V^{\B{m}},\quad \sigma_3W^{\B{m}}\sigma_3=-W^{\B{m}},
\end{equation*}
and hence these matrices take the form
\begin{equation*}
U^{\B{m}}=\begin{pmatrix}
0 & u_1^{\B{m}}\\
u_2^{\B{m}} & 0
\end{pmatrix},\quad V^{\B{m}}=\begin{pmatrix}
v_1^{\B{m}} & 0\\
0 & v_2^{\B{m}}
\end{pmatrix},\quad W^{\B{m}}=\begin{pmatrix}
0 & w_1^{\B{m}}\\
w_2^{\B{m}} & 0
\end{pmatrix}.
\end{equation*}

\begin{lemma}\label{lem:relationvariablesgu}
	The variables $\{g_1^{\B{m}},g_2^{\B{m}},g_3^{\B{m}},g_4^{\B{m}},u_m\}$ and $\{u_1^{\B{m}},u_2^{\B{m}},v_1^{\B{m}},v_2^{\B{m}},w_1^{\B{m}},w_2^{\B{m}}\}$ are completely determined in terms of one and another through the relations
	\begin{align*}
	\kappa(\lambda_m-q^{-1}\lambda_m^{-1})u_1^{\B{m}}&=u_mg_3^{\B{m}},\\
	\kappa(q^{-1}\lambda_m-\lambda_m^{-1})u_2^{\B{m}}&=u_m^{-1}g_4^{\B{m}},\\
	(q^{-2}-1)\kappa \lambda_m v_1^{\B{m}}&=-g_3^{\B{m}}u_2^{\B{m}}+i\lambda_mg_1^{\B{m}},\\
	(q^{-2}-1)\kappa \lambda_m^{-1} v_2^{\B{m}}&=g_4^{\B{m}}u_1^{\B{m}}+i\lambda_m^{-1}g_2^{\B{m}},\\
	\kappa(\lambda_m-q^{-3}\lambda_m^{-1})w_1^{\B{m}}&=u_mg_3^{\B{m}}v_2^{\B{m}}-i\lambda_mg_1^{\B{m}}u_1^{\B{m}}-u_m,\\
	\kappa(q^{-3}\lambda_m-\lambda_m^{-1})w_2^{\B{m}}&=u_m^{-1}g_4^{\B{m}}v_1^{\B{m}}+i\lambda_m^{-1}g_2^{\B{m}}u_2^{\B{m}}+u_m^{-1},
	\end{align*}
	which in particular imply
	\begin{align*}
	u_m=&-i\frac{a_0}{a_1}(1+a_1^2(1+a_2^2))\lambda_m u_1^{\B{m}}+q^{-1}\kappa(q^{-1}\lambda_m-\lambda_m^{-1})u_1^{\B{m}}(u_1^{\B{m}}u_2^{\B{m}}-v_2^{\B{m}})\\
	&+\kappa(\lambda_m-q^{-3}\lambda_m^{-1})w_1^{\B{m}}.
	\end{align*}	
\end{lemma}
\begin{proof}
	Firstly, note that, by equation \eqref{eq:defainf},
	\begin{equation*}
	\Psi_\infty^{\B{m}}(qz)\begin{pmatrix}
	\kappa \lambda_m & 0\\
	0 & \kappa \lambda_m^{-1}
	\end{pmatrix}=z^{-3}A^{\B{m}}(z)\Psi_\infty^{\B{m}}(z),
	\end{equation*}
	Equating the coefficients of $z^{-1},z^{-2}$ and $z^{-3}$ of left and right-hand side gives respectively
	\begin{align*}
	q^{-1}U^{\B{m}}\begin{pmatrix}
	\kappa \lambda_m & 0\\
	0 & \kappa \lambda_m^{-1}
	\end{pmatrix}-\begin{pmatrix}
	\kappa \lambda_m & 0\\
	0 & \kappa \lambda_m^{-1}
	\end{pmatrix}U^{\B{m}}&=A_2^{\B{m}},\\
	q^{-2}V^{\B{m}}\begin{pmatrix}
	\kappa \lambda_m & 0\\
	0 & \kappa \lambda_m^{-1}
	\end{pmatrix}-\begin{pmatrix}
	\kappa \lambda_m & 0\\
	0 & \kappa \lambda_m^{-1}
	\end{pmatrix}V^{\B{m}}&=A_2^{\B{m}} U^{\B{m}}+A_1^{\B{m}},\\
	q^{-3}W^{\B{m}}\begin{pmatrix}
	\kappa \lambda_m & 0\\
	0 & \kappa \lambda_m^{-1}
	\end{pmatrix}-\begin{pmatrix}
	\kappa \lambda_m & 0\\
	0 & \kappa \lambda_m^{-1}
	\end{pmatrix}W^{\B{m}}&=A_2^{\B{m}} V^{\B{m}}+A_1^{\B{m}}U^{\B{m}}+A_0^{\B{m}},
	\end{align*}
	where
	\begin{equation*}
	A^{\B{m}}(z)=A_0^{\B{m}}+zA_1^{\B{m}}+z^2A_2^{\B{m}}+z^3\begin{pmatrix}
	\kappa \lambda_m & 0\\
	0 & \kappa \lambda_m^{-1}
	\end{pmatrix}.
	\end{equation*}
	The desired relations follow directly from these equations.
\end{proof}

In the following proposition we prove part (i) of Theorem \ref{thm:mainresult1}.
\begin{proposition}\label{pro:rhtime}
	Consider the Riemann-Hilbert Problems $\textnormal{RHP}^{\B{m}}(\gamma,C)$, $m\in\mathbb{Z}$, described in Definition \ref{def:RHP1and2}. For every $n\in\mathbb{Z}$, a solution $Y^{\B{m}}(z)$ of $\textnormal{RHP}^{\B{m}}(\gamma,C)$
	exists for at least one $m\in\{n,n+1,n+2\}$. Furthermore, let $m\in\mathbb{Z}$ be such that $Y^{\B{m}}(z)$ exists, then, using the notation in \eqref{eq:widehatpsi}, either
	\begin{enumerate}
		\item[\rm(i\rm)] $u_1^{\B{m}}\neq 0$, in which case $Y^{\B{m+1}}(z)$ exists and
		\begin{align}\label{eq:r+}
		Y^{\B{m+1}}(z)&=R_+^{\B{m}}(z)Y^{m}(z),\\
		R_+^{\B{m}}(z)&=z\begin{pmatrix}
		1 & 0\\
		0 & 0
		\end{pmatrix}+\begin{pmatrix}
		0 & -u_1^{\B{m}}\\
		1/u_1^{\B{m}} & 0
		\end{pmatrix},\label{eq:defiR+}
		\end{align}
		\item[\rm(ii\rm)] $u_1^{\B{m}}=0$, in which case $Y^{\B{m+1}}(z)$ and $Y^{\B{m+2}}(z)$ do not exist, while $Y^{\B{m+3}}(z)$ does exist and
		\begin{align}\label{s+}
		Y^{\B{m+3}}(z)&=S_+^{\B{m}}(z)Y^{m}(z),\\
		S_+^{\B{m}}(z)&=z^3\begin{pmatrix}
		1 & 0\\
		0 & 0
		\end{pmatrix}+z\begin{pmatrix}
		s_m & 0\\
		0 & 0
		\end{pmatrix}+
		\begin{pmatrix}
		0 & -w_1^{\B{m}}\\
		1/w_1^{\B{m}} & 0
		\end{pmatrix},\label{eq:defiS+}\\
		s_m:&=(q^{-2}-1)\frac{\lambda_m v_1^{\B{m}}-q^{-3}\lambda_m^{-1} v_2^{\B{m}}}{\lambda_m-q^{-5}\lambda_m^{-1}}.
		\end{align}
		In particular we necessarily have $w_1^{\B{m}}\neq 0$. 
	\end{enumerate}	
	Similarly either
	\begin{enumerate}
		\item[\rm(iii\rm)] $u_2^{\B{m}}\neq 0$, in which case $Y^{\B{m-1}}(z)$ exists and
		\begin{align}
		Y^{\B{m-1}}(z)&=R_-^{\B{m}}(z)Y^{m}(z),\\
		R_-^{\B{m}}(z)&=z\begin{pmatrix}
		0 & 0\\
		0 & 1
		\end{pmatrix}+\begin{pmatrix}
		0 & 1/u_2^{\B{m}}\\
		-u_2^{\B{m}} & 0
		\end{pmatrix},
		\end{align}
		\item[\rm(iv\rm)] $u_2^{\B{m}}=0$, in which case $Y^{\B{m-1}}(z)$ and $Y^{\B{m-2}}(z)$ do not exist, while $Y^{\B{m-3}}(z)$ does exist and
		\begin{align}
		Y^{\B{m-3}}(z)&=S_-^{\B{m}}(z)Y^{m}(z),\\
		S_-^{\B{m}}(z)&=z^3\begin{pmatrix}
		0 & 0\\
		0 & 1
		\end{pmatrix}+z\begin{pmatrix}
		0 & 0\\
		0 & s_m
		\end{pmatrix}+
		\begin{pmatrix}
		0 & 1/w_2^{\B{m}}\\
		-w_2^{\B{m}} & 0
		\end{pmatrix}.
		\end{align}
		In particular we necessarily have $w_2^{\B{m}}\neq 0$. 
	\end{enumerate}
\end{proposition}
\begin{proof}
	The first statement of the Proposition follows directly from the subsequent four assertions, starting from any seed solution $Y^{\B{m}}(z)$, which is guaranteed to exist by Lemma \ref{lem:existence}. 
	
	Given such a solution $Y^{\B{m}}(z)$, let $R(z)$ be any matrix polynomial and set
	\begin{equation*}
	Y(z)=R(z)Y^{\B{m}}(z),
	\end{equation*}
then $Y(z)$ automatically satisfies the same analyticity and jump condition as $Y^{\B{m}}(z)$. Let $n\in\mathbb{Z}$. If we can choose $R(z)$ such that
\begin{equation}\label{eq:r}
R(z)Y^{\B{m}}(z)=(I+\mathcal{O}\left(z^{-1}\right))z^{n\sigma_3}\quad (z\rightarrow \infty),
\end{equation}
then $Y^{\B{n}}(z)$ exists and 
\begin{equation}\label{eq:ymn}
Y^{\B{n}}(z)=R(z)Y^{\B{m}}(z).
\end{equation} 
Conversely, if $Y^{\B{n}}(z)$ exists, then, defining $R(z)$ by equation \eqref{eq:ymn}, $R(z)$ is a matrix polynomial satisfying \eqref{eq:r}.

To prove the theorem, it remains to study equation \eqref{eq:r}, which can essentially be reduced to linear algebra. Indeed, let us first consider the case $n=m+1$. It is easy to see that $R(z)$ must take the form
 \begin{equation*}
 R_+^{\B{m}}(z)=z\begin{pmatrix}
 1 & 0\\
 0 & 0
 \end{pmatrix}+\begin{pmatrix}
 r_{11}^0 & r_{12}^0\\
 r_{21}^0 & r_{22}^0
 \end{pmatrix},
 \end{equation*}
and we find that, $Y^{\B{m+1}}(z)$ exists if and only if equation \eqref{eq:r} has a solution, which can be rewritten as
 \begin{equation*}
 \begin{pmatrix}
 1 &0 &0 &0\\
 0 & 1 & 0&0\\
 0&0&u_1^{\B{m}}&0\\
 0&0&0&u_1^{\B{m}}
 \end{pmatrix}\begin{pmatrix}
 r_{12}^0\\
 r_{22}^0\\
 r_{11}^0\\
 r_{21}^0
 \end{pmatrix}=\begin{pmatrix}
 -u_1^{\B{m}}\\
 0\\
 0\\
 1
 \end{pmatrix}.
 \end{equation*}
Clearly this system has a solution if and only if $u_1^{\B{m}}\neq 0$, in which case $R(z)=R_+^{\B{m}}(z)$ as defined in equation \eqref{eq:defiR+}.
In particular, if $u_1^{\B{m}}\neq 0$, then $Y^{\B{m+1}}(z)$ indeed exists and equation \eqref{eq:r+} holds true.

Now suppose $u_1^{\B{m}}=0$, then we already know that $Y^{\B{m+1}}(z)$ cannot exist.
As $u_1^{\B{m}}=0$, we have, by Lemma \ref{lem:relationvariablesgu},
\begin{align}
g_1^{\B{m}}&=i(1-q^{-2})\kappa v_1^{\B{m}},\nonumber\\
g_2^{\B{m}}&=i(1-q^{-2})\kappa v_2^{\B{m}},\nonumber\\
g_3^{\B{m}}&=0,\nonumber\\
g_4^{\B{m}}&=\kappa^2(q^{-1}\lambda_m-\lambda_m^{-1})(\lambda_m-q^{-3}\lambda_m^{-1})u_2^{\B{m}}w_1^{\B{m}},\nonumber\\
u_m&=\kappa(\lambda_m-q^{-3}\lambda_m^{-1})w_1^{\B{m}}, \label{eq:w1}
\end{align}
and in particular $w_1^{\B{m}}\neq 0$.
We consider equation \eqref{eq:r} with $n=m+2$. It is easy to see that $R(z)$ must take the form
\begin{equation*}
R(z)=z^2\begin{pmatrix}
1 & 0\\
0 & 0
\end{pmatrix}+z\begin{pmatrix}
r_{11}^1 & r_{12}^1\\
r_{21}^1 & r_{22}^1
\end{pmatrix}+\begin{pmatrix}
r_{11}^0 & r_{12}^0\\
r_{21}^0 & r_{22}^0
\end{pmatrix},
\end{equation*}
and we find that, $Y^{\B{m+2}}(z)$ exists, if and only if equation \eqref{eq:r} has a solution, which can be rewritten as
\begin{equation*}
\begin{pmatrix}
1 &0 &0 &0          &0&0&0&0\\
0 & 1 & 0&0         &0&0&0&0 \\
0&0&1&0             &0&0&0&0\\
0&0&0&1             &0&0&0&0\\
v_2^{\B{m}} &0&0&0       &0&0&0&0\\
0 &v_2^{\B{m}}&0&0      &0&0&0&0\\
0 &0&v_2^{\B{m}}&0      &w_1^{\B{m}}&0&0&0\\
0 &0&0&v_2^{\B{m}}      &0&w_1^{\B{m}}&0&0
\end{pmatrix}\begin{pmatrix}
r_{12}^1\\
r_{22}^1\\
r_{12}^0\\
r_{22}^0\\
r_{11}^1\\
r_{21}^1\\
r_{11}^0\\
r_{21}^0
\end{pmatrix}=\begin{pmatrix}
-u_1^{\B{m}}\\
0\\
0\\
0\\
-w_1^{\B{m}}\\
0\\
0\\
0
\end{pmatrix}.
\end{equation*}
As $w_1^{\B{m}}\neq 0$, but $u_1^{\B{m}}=0$, this equation does not have a solution and hence $Y^{\B{m+2}}(z)$ cannot exist.

We now show the existence of $Y^{\B{m+3}}(z)$. We consider equation \eqref{eq:r}, with $R(z)=S(z)$ and $n=m+3$. To simplify the procedure, note that, by equation \eqref{eq:psiinfsym}, we must have
\begin{equation*}
S(-z)=-\sigma_3S(z)\sigma_3,
\end{equation*}
and hence any solution $S(z)$ must take the form
\begin{equation*}
S(z)=z^3\begin{pmatrix}
1 & 0\\
0 & 0
\end{pmatrix}+z^2\begin{pmatrix}
0 & s_1^2\\
s_2^2 & 0
\end{pmatrix}+z\begin{pmatrix}
s_1^1 & 0\\
0 & s_2^1
\end{pmatrix}+
\begin{pmatrix}
0 & s_1^0\\
s_2^0 & 0
\end{pmatrix}.
\end{equation*}
We extend \eqref{eq:widehatpsi}, by writing
\begin{align*}
\Psi_\infty^{\B{m}}(z)&=I+z^{-1}U^{\B{m}}+z^{-2}V^{\B{m}}+z^{-3}W^{\B{m}}+z^{-4}X^{\B{m}}+z^{-5}Z^{\B{m}}+\mathcal{O}\left(z^{-6}\right)\quad (z\rightarrow \infty),
\end{align*}
where $X^{\B{m}}$ and $Z^{\B{m}}$ take the form
\begin{equation*}
X^{\B{m}}=\begin{pmatrix}
x_1^{\B{m}} & 0\\
0 & x_2^{\B{m}}
\end{pmatrix},\quad Z^{\B{m}}=\begin{pmatrix}
0 & z_1^{\B{m}}\\
z_2^{\B{m}} & 0
\end{pmatrix}.
\end{equation*}
Then equation \eqref{eq:r} with $n=m+3$ for $R(z)=S(z)$, is equivalent to
\begin{equation*}
\begin{pmatrix}
	1 &0&0&0&0&0\\
	0 &1&0&0&0&0\\
	v_2^{\B{m}} &0&1&0&0&0\\
	0&v_2^{\B{m}}&0&w_1^{\B{m}}&0&0\\
	x_2^{\B{m}}&0&v_2^{\B{m}}&0&w_1^{\B{m}}&0\\
	0&x_2^{\B{m}}&0&z_1^{\B{m}}&0&w_1^{\B{m}}\\
\end{pmatrix}\cdot
\begin{pmatrix}
s_1^2\\
s_2^1\\
s_1^0\\
s_2^2\\
s_1^1\\
s_2^0
\end{pmatrix}=
\begin{pmatrix}
0\\
0\\
-w_1^{\B{m}}\\
0\\
-z_1^{\B{m}}\\
1
\end{pmatrix}.
\end{equation*}
We know $w_1^{\B{m}}\neq 0$, by equation \eqref{eq:w1}, which implies that the above equation has a unique solution, given by
\begin{equation*}
s_1^2=0,\quad s_2^2=0,\quad s_1^1=v_2^{\B{m}}-z_1^{\B{m}}/w_1^{\B{m}},\quad s_2^1=0,\quad s_1^0=-w_1^{\B{m}},\quad s_2^0=1/w_1^{\B{m}},
\end{equation*}
and hence $Y^{\B{m+3}}(z)=S(z)Y^{\B{m}}(z)$ exists.
It remains to be checked that
\begin{equation}\label{eq:identityy1}
v_2^{\B{m}}-z_1^{\B{m}}/w_1^{\B{m}}=s_m.
\end{equation}
Using the same method as in the proof of Lemma \ref{lem:relationvariablesgu}, we find
\begin{equation*}
z_1^{\B{m}}=(\lambda_m-q^{-5}\lambda_m^{-1})^{-1}\left[(1-q^{-2})\lambda_m v_1^{\B{m}}+(\lambda_m-q^{-3}\lambda_m^{-1})v_2^{\B{m}}\right]w_1^{\B{m}},
\end{equation*}
from which equation \eqref{eq:identityy1} follows directly. We conclude that expression \eqref{eq:defiS+} is indeed correct.
 The second part of the theorem is proven analogously.
\end{proof}

\begin{corollary}\label{cor:timeevolutiona}
	Considering the generalised inverse monodromy problem \ref{prob:inversem}, for every $n\in\mathbb{Z}$, the solution $A^{\B{m}}(z)$
	exists for at least one $m\in\{n,n+1,n+2\}$. Furthermore, let $m\in\mathbb{Z}$ be such that $A^{\B{m}}(z)$ exists, then, using the notation in \eqref{eq:Ameq}, either
\begin{enumerate}
	\item[\rm(i\rm)] $g_3^{\B{m}}\neq 0$, in which case $A^{\B{m+1}}(z)$ exists and equals
		\begin{align}
		&\mathcal{A}(z;q\lambda_m,g^{\B{m+1}},u_{m+1})=R_+^{\B{m}}(qz)\mathcal{A}(z;\lambda_m,g^{\B{m}},u_m)R_+^{\B{m}}(z)^{-1},\label{eq:time-evo_A}\\
		&R_+^{\B{m}}(z)=z\begin{pmatrix}
		1 & 0\\
		0 & 0
		\end{pmatrix}+\begin{pmatrix}
		0 & -u_1^{\B{m}}\\
		1/u_1^{\B{m}} & 0
		\end{pmatrix},\nonumber\\
		&u_1^{\B{m}}=\kappa(\lambda_m-q^{-1}\lambda_m^{-1})^{-1}u_mg_3^{\B{m}}.\nonumber
		\end{align}
	\item[\rm(ii\rm)] $g_3^{\B{m}}=0$, in which case $A^{\B{m+1}}(z)$ and $A^{\B{m+2}}(z)$ do not exist whereas $A^{\B{m+3}}(z)$ does exist and equals
	\begin{align*}
	&\mathcal{A}(z;q^3\lambda_m,g^{\B{m+3}},u_{m+3})=S_+^{\B{m}}(qz)\mathcal{A}(z;\lambda_m,g^{\B{m}},u_m)S_+^{\B{m}}(z)^{-1},\\
	&S_+^{\B{m}}(z)=z^3\begin{pmatrix}
	1 & 0\\
	0 & 0
	\end{pmatrix}+z\begin{pmatrix}
	s_m & 0\\
	0 & 0
	\end{pmatrix}+
	\begin{pmatrix}
	0 & -w_1^{\B{m}}\\
	1/w_1^{\B{m}} & 0
	\end{pmatrix},\\
	&w_1^{\B{m}}=\kappa(\lambda_m^{-1}-q^{-3}\lambda_m^{-1})^{-1}u_m.
	\end{align*}
\end{enumerate}
Here $s_m$ is given by
\begin{equation*}
s_m=i\kappa^{-1}(\lambda_m-q^{-5}\lambda_m^{-1})^{-1}(\lambda_mg_1^{\B{m}}-q^{-3}\lambda_m^{-1}g_2^{\B{m}}).
\end{equation*}
Similarly, either
\begin{enumerate}
	\item[\rm(iii\rm)] $g_4^{\B{m}}\neq 0$, in which case $A^{\B{m-1}}(z)$ exists and equals
		\begin{align*}
		&\mathcal{A}(z;q^{-1}\lambda_m,g^{\B{m-1}},u_{m-1})(z)=R_-^{\B{m}}(qz)\mathcal{A}(z;\lambda_m,g^{\B{m}},u_m)R_-^{\B{m}}(z)^{-1},\\
			&R_-^{\B{m}}(z)=z\begin{pmatrix}
			0 & 0\\
			0 & 1
			\end{pmatrix}+\begin{pmatrix}
			0 & 1/u_2^{\B{m}}\\
			-u_2^{\B{m}} & 0
			\end{pmatrix},\\
		&u_2^{\B{m}}=\kappa(q^{-1}\lambda_m-\lambda_m^{-1})^{-1}u_m^{-1}g_4^{\B{m}}.
		\end{align*}
	\item[\rm(iv\rm)] $g_4^{\B{m}}=0$, in which case $A^{\B{m-1}}(z)$ and $A^{\B{m-2}}(z)$ do not exist whereas $A^{\B{m-3}}(z)$ does exist and equals
	\begin{align*}
	&\mathcal{A}(z;q^{-3}\lambda_m,g^{\B{m-3}},u_{m-3})(z)=S_-^{\B{m}}(qz)\mathcal{A}(z;\lambda_m,g^{\B{m}},u_m)S_-^{\B{m}}(z)^{-1},\\
&S_-^{\B{m}}(z)=z^3\begin{pmatrix}
0 & 0\\
0 & 1
\end{pmatrix}+z\begin{pmatrix}
0 & 0\\
0 & s_m
\end{pmatrix}+
\begin{pmatrix}
0 & 1/w_2^{\B{m}}\\
-w_2^{\B{m}} & 0
\end{pmatrix},\\
	&w_2^{\B{m}}=\kappa(q^{-3}\lambda_m-\lambda_m^{-1})^{-1}u_m^{-1}.
	\end{align*}
\end{enumerate}
\end{corollary}
\begin{proof}
	This follows directly from Proposition \ref{pro:rhtime} and Lemma \ref{lem:relationvariablesgu}, using the equivalence of the generalised inverse monodromy problem \ref{prob:inversem} and $\text{RHP}^{\B{m}}(\gamma,C)$, $m\in\mathbb{Z}$.
\end{proof}
\begin{remark}
	Note that, not coincidently, equation \eqref{eq:time-evo_A} agrees perfectly with Equation \eqref{eq:laxtime}, namely $R_+^{\B{m}}(z)={\mathcal B}(z;\lambda_m,f^{\B{m}},u_m)$.
\end{remark}
Finally the following lemma shows that the isomonodromic deformation of the class of Fuchsian systems \ref{eq:qdifferenceclass} as $\lambda\rightarrow q\lambda$ is equivalent to the $q\Pfour^{\text{mod}}$ time-evolution.
\begin{lemma}\label{lem:equivalencetime}
The time-evolution of $g^{\B{m}}\in\mathcal{G}(a)$ and $u_m\in\mathbb{C}^*$, induced by Corollary \ref{cor:timeevolutiona}, coincides with $q\Pfour^{\text{mod}}(\lambda_0,a)$ plus its continuation formulae \eqref{eq:gcontinuation}, and the auxiliary equation \eqref{eq:auxiliary}.
\end{lemma}
\begin{proof}
This follows by direct calculation.
\end{proof}

We now have all the ingredients to prove Theorem \ref{thm:mainresult1} in Section \ref{s:mr}.
\begin{proof}[Proof of Theorem \ref{thm:mainresult1}]
Firstly, note that part (i) follows from Proposition \ref{pro:rhtime}. As to part (ii), observe that the definition of $A^{\B{m}}(z)$ coincides with the one in Proposition \ref{cor:imrhequivalencem}, i.e. equation \eqref{eq:defainfm}. So indeed $A^{\B{m}}(z)\in \mathcal{F}(q^m\lambda_0,a)$, by (ii) in Proposition \ref{cor:imrhequivalencem}.
	
Finally part (iii) follows from Corollary \ref{cor:timeevolutiona} and Lemma \ref{lem:equivalencetime}.
\end{proof}

\subsection{Bijectivity of the Riemann-Hilbert Mapping}
\label{section:rhcorrespondence}
In this section we prove Theorem \ref{thm:rhmap}. 
\begin{proof}[Proof of Theorem \ref{thm:rhmap}]
Note that Theorem \ref{thm:mainresult1} allows us to associate with any connection matrix $C(z)\in \mathcal{C}(\lambda_0,a)$, a unique $q\Pfour^{\text{mod}}(\lambda_0,a)$ transcendent $g=(g^{\B{m}})_{m\in\mathbb{Z}}$ and solution $u=(u_m)_{m\in\mathbb{Z}}$ of the auxiliary equation \eqref{eq:auxiliary}.
Upon scaling $C(z)\rightarrow \widehat{C}(z)=C(z)D$, where $D=\operatorname{diag}(d_1,d_2)$ an invertible diagonal matrix, the solution of $\text{RHP}^{\B{m}}(\gamma,C)$ is scaled by $Y^{\B{m}}\rightarrow \widehat{Y}^{\B{m}}$, where
\begin{equation}\label{eq:scaling}
\widehat{Y}^{\B{m}}(z)=\begin{cases}
D^{-1}Y^{\B{m}}(z)D & \text{if $z\in D_+$,}\\
D^{-1}Y^{\B{m}}(z) & \text{if $z\in D_-$.}
\end{cases}
\end{equation}
In turn this scales the matrix $A^{\B{m}}(z)$ to $\widehat{A}^{\B{m}}(z)=D^{-1}A^{\B{m}}(z)D$, leaving the underlying $q\Pfour^{\text{mod}}(\lambda_0,a)$ transcendent $g=\widehat{g}$ invariant whilst rescaling the solution of the auxiliary equation by $u\rightarrow \widehat{u}=\frac{d_2}{d_1}u$. Therefore the Riemann-Hilbert mapping \ref{eq:RHmapping} is well-defined.
It also follows that both $g$ and $u$ remain invariant under scaling $C(z)\rightarrow \widehat{C}(z)=c\,C(z)$, for any $c\in\mathbb{C}^*$, yielding the mapping
\begin{equation*}
RH:M_c(\lambda_0,a)\rightarrow S(\lambda_0,a),\quad M \overset{RH}{\mapsto}(g,u),
\end{equation*}
where $M_c(\lambda_0,a)$ is defined in Definition \ref{def:monodromymap} and $S(\lambda_0,a)$ denotes the solution space of $q\Pfour^{\text{mod}}(\lambda_0,a)$ plus the auxiliary equation \eqref{eq:auxiliary}. Furthermore, for $M\in M_c(\lambda_0,a)$, if $RH(M)=(g,u)$, then
\begin{equation*}
RH(M\cdot \operatorname{diag}(1,d))=(g,du),\quad M\cdot \operatorname{diag}(1,d):=\{C(z)\cdot\operatorname{diag}(1,d):C(z)\in M\}.
\end{equation*}
Thus, to prove the theorem, it suffices to show that the mapping $RH$ is bijective.

We proceed with constructing an inverse of $RH$. Let $g=(g^{\B{m}})_{m\in\mathbb{Z}}$ be any $q\Pfour^{\text{mod}}(\lambda_0,a)$ transcendent and $u=(u_m)_{m\in\mathbb{Z}}$ be any solution of the auxiliary equation. Denote by $X\subseteq Z$ the set of integers $m$ where $g^{\B{m}}$ is singular, i.e. $g^{\B{m}}=\mathsf{s}$, recalling Definition \ref{def:transcendent}.

For $m\in \mathbb{Z}\setminus X$, we write
\begin{equation}
A^{\B{m}}(z):=\mathcal{A}(z;\lambda_m,g^{\B{m}},u_m),
\end{equation}
denote corresponding monodromy by
\begin{equation*}
M_m^*:=M_{\mathcal{F}}(\mathcal{A}^{\B{m}}(z))\in M_c(q^m\lambda_0,a),
\end{equation*}
and set
\begin{equation*}
M_m:=\tau^{-m}(M_m^*)\in M_c(\lambda_0,a),
\end{equation*}
recalling the definition of $\tau$ in equation \eqref{eq:taudef}.
Then we know, by Corollary \ref{cor:timeevolutiona} and Lemma \ref{lem:equivalencetime}, that the monodromy $M_m=M$ is independent of $m\in\mathbb{Z}\setminus X$. We write $M_\text{IV}(g,u)=M$, yielding a mapping
\begin{equation}\label{eq:miv}
M_\text{IV}:S(\lambda_0,a)\rightarrow M_c(\lambda_0,a).
\end{equation}
Due to Lemma \ref{lem:equivalencetime} and the equivalence of the generalised inverse monodromy problem \ref{prob:inversem} and the main RHP defined in Definition \ref{def:RHP1and2}, see Proposition \ref{cor:imrhequivalencem}, it is evident that $M_\text{IV}$ is an inverse of the mapping $RH$. In particular $RH$ is bijective and the theorem follows.
\end{proof}

%% file: modulispace.tex
\section{The moduli space}\label{section:moduli}
In this section we study the monodromy surface defined in Definition \ref{def:modulispace}. In Section \ref{subsec:proof} we prove Theorem \ref{thm:mainresult3} and in Section \ref{subsec:real} we classify those monodromy data corresponding to real-valued transcendents, yielding Remark \ref{rem:real}.

\subsection{Proof of Theorem \ref{thm:mainresult3}}\label{subsec:proof}
In order to study the monodromy surface $\mathcal{M}_c(\lambda,a)$, defined in \ref{def:modulispace}, we briefly recall some fundamental properties of theta functions relevant in this context. Namely we consider analytic functions $c(z)$ on $\mathbb{C}^*$, such that $c(z)/c(qz)$ is a monomial. We call such functions $q$-theta functions.

For $\alpha\in\mathbb{C}^*$ and $n\in\mathbb{N}$, we denote by $V_n(\alpha)$ the set of all analytic functions $c(z)$ on $\mathbb{C}^*$, satisfying
\begin{equation}\label{eq:theta}
c(qz)=\alpha z^{-n}c(z).
\end{equation}
We note that $V_n(\alpha)$ is a vector space of dimension $n$ if $n\geq 1$ \cite{phdroffelsen}.

 Recalling Definition \ref{def:connectionspace}, note that for any connection matrix $C(z)\in \mathfrak{C}(\lambda,a)$, all of its entries are $q$-theta functions. 
 For $r\in\mathbb{R}_+$, we call
\begin{equation*}
	D_q(r):=\{|q|r\leq|z|<r\},
\end{equation*}
a fundamental annulus. $q$-theta functions are, up to scaling, completely determined by the location of their zeros within any fixed fundamental annulus. Indeed, we have the following
\begin{lemma}\label{lem:classificationtheta}
	Let $\alpha\in\mathbb{C}^*$, $n\in\mathbb{N}$ and  $c(z)\not\equiv 0$ be an element of $V_n(\alpha)$. Then, within any fixed fundamental annulus, $c(z)$ has precisely $n$ zeros, counting multiplicity, say $\{a_1,\ldots,a_n\}$, and there exist unique $c\in\mathbb{C}^*$ and $s\in\mathbb{Z}$ such that
	\begin{equation}\label{eq:thetaformula}
	c(z)=c z^s\theta_q(z/a_1,\ldots,z/a_{n}),\quad \alpha=(-1)^n q^s a_1\cdot\ldots\cdot a_n.
	\end{equation}	
	Conversely, for any choice of the parameters equation \eqref{eq:thetaformula} defines an element of $V_n(\alpha)$.
\end{lemma}
\begin{proof}
	See for instance \cite{phdroffelsen}.
\end{proof}

\begin{proof}[Proof of Theorem \ref{thm:mainresult3}]
We first prove that the mapping $\rho$ is injective and then prove that its range equals the algebraic surface $\mathcal{P}(\lambda_0,a)$.

Let $M=[C],\widetilde{M}=[\widetilde{C}]\in \mathcal{M}_c(\lambda_0,a)$ and suppose that corresponding coordinates $\rho_{1,2,3}$ and $\widetilde{\rho}_{1,2,3}$ are equal. Set $D(z)=C(z)^{-1}\widetilde{C}(z)$, then $D(z)$ is a meromorphic function on $\mathbb{C}^*$ satisfying
\begin{equation}\label{eq:qR}
D(qz)=\lambda_0^{\sigma_3} D(z)\lambda_0^{-\sigma_3}.
\end{equation}
We know that $D(z)$ is analytic away from the $q$-spirals $\pm q^\mathbb{Z}x_k$, $1\leq k\leq 3$. Let $1\leq k\leq 3$, then $\pi(\widetilde{C}(x_k))=\pi(C(x_k))$ and thus $\pi(\widetilde{C}(q^mx_k))=\pi(C(q^mx_k))$ for all $m\in\mathbb{Z}$, which implies that $D(z)$ is analytic at the points on the $q$-spiral $+q^\mathbb{Z}x_k$. Furthermore, due to the symmetry \eqref{item:c4} in Definition \ref{def:connectionspace}, we have $D(-z)=\sigma_3 D(z)\sigma_3$ and thus $D(z)$ is also analytic at the points on the $q$-spiral $-q^\mathbb{Z}x_k$. We conclude that $D(z)$ is analytic on $\mathbb{C}^*$. It thus follows immediately from Lemma \ref{lem:classificationtheta} and equation \eqref{eq:qR} that
\begin{equation*}
D(z)\equiv\begin{pmatrix}
d_1 & 0\\
0 & d_2
\end{pmatrix},
\end{equation*}
for some constants $d_1,d_2\in\mathbb{C}^*$, since $\lambda_0^2\notin q^\mathbb{Z}$. Therefore $\widetilde{C}=CD$ and $C$ lie in the same equivalence class in $\mathcal{M}_c(\lambda,a)$, i.e. $\widetilde{M}=M$. It follows that the mapping $\rho$ is injective.

Next, we show that the range of $\rho$ is contained in the algebraic surface $\mathcal{P}(\lambda_0,a)$. Let
$M=[C]\in \mathcal{M}_c(\lambda_0,a)$, then,
because of the symmetry \eqref{item:c3}, $C(z)$ is of the form
\begin{equation}\label{eq:C12}
	C(z)=\begin{pmatrix}
		C_1(z) & C_2(z)\\
		-C_1(-z) & C_2(-z)
	\end{pmatrix}.
\end{equation}
Due to \eqref{item:c1},\eqref{item:c2} and \eqref{item:c4}, $C_1(z)$ and $C_2(z)$ are $q$-theta functions lying in the respective spaces
\begin{align*}
U&:=V_3((qa_0^2a_2)^{-1}\lambda_0^{-1}),\\
V&:=V_3((qa_0^2a_2)^{-1}\lambda_0^{+1}),
\end{align*}
 satisfying the identity
\begin{align}\label{eq:Cdet}
&C_1(z)C_2(-z)+C_1(-z)C_2(z)=cw(z),\\
 &w(z):=\theta_q(+z/x_1,-z/x_1,+z/x_2,-z/x_2,+z/x_3,-z/x_3).\nonumber
\end{align}
for some $c\in\mathbb{C}^*$.\\
To proceed, we fix explicit bases of $U$ and $V$. Define
\begin{align*}
u_1(z)&=\theta_q(z/x_2,z/x_3,-z/x_1\lambda_0),& v_1(z)&=\theta_q(z/x_2,z/x_3,-z/x_1\lambda_0^{-1}),\\
u_2(z)&=\theta_q(z/x_1,z/x_3,-z/x_2\lambda_0),& v_2(z)&=\theta_q(z/x_1,z/x_3,-z/x_2\lambda_0^{-1}),\\
u_3(z)&=\theta_q(z/x_1,z/x_2,-z/x_3\lambda_0),& v_3(z)&=\theta_q(z/x_1,z/x_2,-z/x_3\lambda_0^{-1}).
\end{align*}
Then $\{u_1,u_2,u_3\}$ is a basis of $U$ and $\{v_1,v_2,v_3\}$ is a basis of $V$. We have chosen these bases such that $u_k(x_l)=v_k(x_l)=0$ if $k\neq l$ for $1\leq k,l\leq 3$.

Now, let $\alpha\in\mathbb{C}^3$ and $\beta\in\mathbb{C}^3$ be such that
\begin{align}
C_1(z)&=\alpha_1u_1(z)+\alpha_2u_2(z)+\alpha_3u_3(z),\label{eq:alpha}\\
C_2(z)&=\beta_1v_1(z)+\beta_2v_2(z)+\beta_3v_3(z).\label{eq:beta}
\end{align}

We now proceed to show that $\rho=\rho(M)$ is an element of the algebraic surface $\mathcal{P}(\lambda_0,a)$.
We use the standard notation $p=[p^x:p^y]$ for elements $p\in \mathbb{P}^1$ and accordingly write  $\rho_k=[\rho_k^x:\rho_k^y]$ for $1\leq k\leq 3$.
Let $1\leq k \leq 3$, then $\pi(C(x_k))=\rho_k$ implies $\rho_k^yC_2(x_k)=\rho_k^xC_2(-x_k)$ and thus
\begin{equation}\label{eq:betahomogen}
\rho_k^yv_k(x_k)\beta_k=\rho_k^x(\beta_1v_1(-x_k)+\beta_2v_2(-x_k)+\beta_3v_3(-x_k)),
\end{equation}
for $k\in \{1,2,3\}$. This is a homogeneous linear system in $\beta_{1,2,3}$. Since $\beta$ is nonzero, this implies
\begin{equation}\label{eq:algebraicconstraint}
\Delta_1:=\begin{vmatrix}
\rho_1^x v_1(-x_1)-\rho_1^y v_1(x_1) & \rho_1^x v_2(-x_1) & \rho_1^x v_3(-x_1)\\
\rho_2^x v_1(-x_2) & \rho_2^x v_2(-x_2)-\rho_2^y v_2(x_2) & \rho_2^x v_3(-x_2)\\
\rho_3^x v_1(-x_3) & \rho_3^x v_2(-x_3) & \rho_3^x v_3(-x_3)-\rho_3^y v_3(x_3)
\end{vmatrix}=0.
\end{equation}
 This equation is precisely
 \begin{equation*}
     T_{hom}(\rho_1^x,\rho_1^y,\rho_2^x,\rho_2^y,\rho_3^x,\rho_3^y;\lambda_0,a)=0,
 \end{equation*}
after some simplification. (Details can be found in Appendix \ref{append:derivationcubic}; see equations \eqref{eq:determinantequality} and \eqref{eq:determinantsimple}.)
It follows that indeed the range of $\rho$ is contained in $\mathcal{P}(\lambda_0,a)$.

To complete the proof, it remains to be shown that any element of $\mathcal P(\lambda_0,a)$ can be realised as the coordinates of an equivalence class $M=[C(z)]\in \mathcal{M}_c(\lambda_0,a)$.

Take any $(\rho_1,\rho_2,\rho_3)\in \mathcal{P}(\lambda_0,a)$. Then we know that the determinant in \eqref{eq:algebraicconstraint} vanishes. Thus there exists a nonzero solution $\beta\in \mathbb{C}^3$ of the homogeneous linear system \eqref{eq:betahomogen}.

We define $C_2(z)$ by equation \eqref{eq:beta},
then $C_2\in V$ and 
\begin{equation}\label{eq:C2cond}
\rho_k^yC_2(x_k)=\rho_k^xC_2(-x_k),
\end{equation}
for $1\leq k\leq 3$. Similarly, to ensure that $\pi(C(x_k))=\rho_k$, we must have 
\begin{equation}\label{eq:C1cond}
\rho_k^yC_1(x_k)=-\rho_k^xC_1(-x_k),
\end{equation}
which, using the notation in \eqref{eq:alpha},
is equivalent to
\begin{equation}\label{eq:alphaeq}
-\rho_k^yu_k(x_k)\alpha_k=\rho_k^x(\alpha_1u_1(-x_k)+\alpha_2u_2(-x_k)+\alpha_3u_3(-x_k)),
\end{equation}
for $k\in \{1,2,3\}$.
This homogeneous linear system has a nonzero solution $\alpha\in\mathbb{C}^3$ if and only if the determinant
\begin{equation*}\label{eq:detalpha}
\Delta_2=\begin{vmatrix}
	\rho_1^x u_1(-x_1)+\rho_1^y u_1(x_1) & \rho_1^x u_2(-x_1) & \rho_1^x u_3(-x_1)\\
	\rho_2^x u_1(-x_2) & \rho_2^x u_2(-x_2)+\rho_2^y u_2(x_2) & \rho_2^x u_3(-x_2)\\
	\rho_3^x u_1(-x_3) & \rho_3^x u_2(-x_3) & \rho_3^x u_3(-x_3)+\rho_3^y u_3(x_3)
\end{vmatrix}
\end{equation*}
vanishes. Direct computation gives that this determinant equals $\Delta_2=-\lambda_0^3\Delta_1$, where $\Delta_1$ is the determinant in \eqref{eq:algebraicconstraint}; see equation \eqref{eq:determinantequality} in Appendix \ref{append:derivationcubic}. Since $\Delta_1=0$, the linear system \eqref{eq:alphaeq} has a nonzero solution $\alpha\in\mathbb{C}^3$, and with this choice of $\alpha$ in \eqref{eq:alpha}, we know that $C_1\in U$ satisfies equation \eqref{eq:C1cond}.

Define the matrix function $C(z)$ by equation \eqref{eq:C12}, then, by construction, it satisfies properties \eqref{item:c1}, \eqref{item:c2} and \eqref{item:c4}. It only remains to be checked that equation \eqref{eq:Cdet} holds true. To this end, let us write
\begin{equation*}
W(z):=|C(z)|=C_1(z)C_2(-z)+C_1(-z)C_2(z).
\end{equation*}
Then $W(z)$, just like $w(z)$, is a $q$-theta function which lies in $V_6(-(qa_0^2a_2)^{-2})$. Thus, to show that equation \eqref{eq:Cdet} holds, all we have to do is check that $W(z)$ and $w(z)$ have the same zeros, due to Lemma \ref{lem:classificationtheta}. Namely we have to check that $W(\pm x_k)=0$ for $1\leq k\leq 3$. However the latter follows trivially from equations \eqref{eq:C1cond} and \eqref{eq:C2cond}. We conclude that $C(z)\in\mathfrak{C}(\lambda_0,a)$, and $\pi(C(x_k))=\rho_k$, $1\leq k\leq3$, due to equations \eqref{eq:C1cond} and \eqref{eq:C2cond}. The theorem follows.
\end{proof}

\subsection{Real-valued Transcendents}\label{subsec:real}
In this section we characterise those monodromy data which yield real solutions. Take $a\in\mathbb{R}^3$ and $\lambda_0\in \mathbb{R}^*$  such that $q=a_0a_1a_2\in (-1,1)\setminus\{0\}$ and the non-resonant conditions \eqref{eq:conditionnonresonant} are satisfied.
Then a $q\Pfour(\lambda_0,a)$ transcendent $f$ is real-valued, if and only if its associated $q\Pfour^{\text{mod}}(\lambda_0,a)$ transcendent $g$ (given by \eqref{eq:aginf}) is real-valued.

If $g$ is real-valued, then we can choose a solution $u$ of the auxiliary equation which is purely imaginary.
Then the corresponding matrix
\begin{equation}
A^{\B{m}}(z):=\mathcal{A}(z;q^m\lambda_0,g^{\B{m}},u_m),
\end{equation}
satisfies
\begin{equation}\label{eq:Areal}
\overline{A^{\B{m}}(\overline{z})}=-A^{\B{m}}(z).
\end{equation}
It follows that the fundamental solutions $\Phi_0^{\B{m}}(z)$ and $\Phi_\infty^{\B{m}}(z)$, defined in Lemmas \ref{lem:solzero} and \ref{lem:solinf}, with $d\in\mathbb{R}$, are real analytic.
Thus the corresponding connection matrix $C(z)$ is real analytic, that is
\begin{equation}\label{eq:Creal}
\overline{C(\overline{z})}=C(z).
\end{equation}

Conversely, suppose $C(z)\in \mathfrak{C}(\lambda,a)$ is real analytic. Choose an admissible Jordan curve $\gamma$  such that $\overline{\gamma}=\gamma$,
then the solution $Y^{\B{m}}(z)$ of $\text{RHP}^{\B{m}}(\gamma,C)$ satisfies $\overline{Y^{\B{m}}(\overline{z})}=-Y^{\B{m}}(z)$ for $z\in\mathbb{C}\setminus \gamma$. Therefore $A^{\B{m}}(z)$ satisfies \eqref{eq:Areal} from which it follows that $g$ and $f$ are real-valued.

We conclude that a monodromy datum $M\in\mathcal{M}_c(\lambda,a)$ corresponds to a real solution $f$, via the Riemann-Hilbert mapping in Theorem \ref{thm:rhmap}, if and only if there exists a representative $C(z)\in M$ which is real analytic. In turn it is easy to see that the latter holds true if and only if $\rho(M)\in \mathbb P^1_{\mathbb R}\times \mathbb P^1_{\mathbb R}\times \mathbb P^1_{\mathbb R}$. Indeed, the forward implication is trivial and its converse follows from the fact that, if $\rho\in \mathbb P^1_{\mathbb R}\times \mathbb P^1_{\mathbb R}\times \mathbb P^1_{\mathbb R}$, then the homogeneous linear systems \eqref{eq:alphaeq} and \eqref{eq:betahomogen} have real nonzero solutions $\alpha\in\mathbb{R}^3$ and $\beta\in\mathbb{R}^3$ respectively. Remark \ref{rem:real} follows.

%% file: conclusion.tex
\section{Conclusion}\label{s:con}
In this paper, we derived a Riemann-Hilbert representation for the general solution of $q\Pfour$ in the non-resonant parameter case. Most importantly, we showed the bijection of the mapping that associates to any $q\Pfour$ transcendent a corresponding equivalence class of connection matrices in the monodromy surface. Furthermore, we constructed an explicit algebraic surface, which is the moduli space of the monodromy surface and thus of $q\Pfour$.

This lays the groundwork for the global asymptotic analysis of solutions of $q\Pfour$. In particular, by studying the Riemann-Hilbert problem $\text{RHP}^{\B{m}}(\gamma,C)$ in the limits $m\rightarrow +\infty$ and $m\rightarrow -\infty$, we plan to derive corresponding asymptotic behaviours for solutions of $q\Pfour$ and associated connection formulae, analogous to the differential theory \cites{fokas,kapaev1996}, in a forthcoming paper.

An interesting question concerns how the Riemann-Hilbert theory developed here will extend to the resonant regime, previously studied by Adams \cite{adams1928}. We intend to use our approach for this case to study special solutions, as has been done for the differential fourth Painlev\'e equation \cites{kapaev1998,buckinghampiv,masoeroroffelsen,masoeroroffelsen2}.

Finally, an intriguing question remains open on whether our Riemann-Hilbert representation of $q\Pfour$ can be used to derive limiting results in the continuum limit $q\rightarrow 1$.

%% file: birational.tex
\section{A birational transformation and  singularities} \label{appendix:singularityconfinement}
Define
\begin{subequations}\label{eq:aginf}
	\begin{align}
	g_1&=qf_2^{-1}+a_0a_2f_1^{-1}f_2^{-1}+a_0f_1^{-1},\\
	g_2&=qf_2+a_0a_2f_1f_2+a_0f_1,\\
	g_3&=qa_0a_2f_1+qa_0f_1f_2^{-1}+a_0^2a_2f_2^{-1},\\
	g_4&=qa_0a_2f_1^{-1}+qa_0f_2f_1^{-1}+a_0^2a_2f_2,
	\end{align}
\end{subequations}
then $g=(g_1,g_2,g_3,g_4)$ satisfies the algebraic equations \eqref{eq:algebraic}
and the rational inverse of \eqref{eq:aginf} is given by
\begin{subequations}\label{eq:afing}
	\begin{align}
	f_1&=\frac{a_0^2+g_3}{a_0(qa_2+g_1)}=\frac{a_0(qa_2+g_2)}{a_0^2+g_4},\\
	f_2&=\frac{q^2+g_4}{a_0a_2(a_0+a_1g_1)}=\frac{a_0a_2(a_0+a_1g_2)}{q^2+g_3}.
	\end{align}
\end{subequations}
We denote the algebraic surface obtained by cutting $\{g\in\mathbb{C}^4\}$ with respect to \eqref{eq:algebraic} by $\mathcal{G}(a)$.
The $f$ and $g$ variables are bi-rationally equivalent, and in particular $q\Pfour(a)$ induces the time-evolution given by Equation \eqref{eq:aqp4} on $\mathcal{G}(a)$.

 While the forward iteration of Equation \eqref{eq:aqp4} is singular on $\mathcal{G}(a)$, only when $g_3=0$, we show its continuation is possible by means of singularity confinement. It is also possible to regularize these singularities by lifting to the initial value space $(A_2+A_1)^{(1)}$ following Sakai \cite{sakai2001}. 

Namely, if $g_3=0$, then $\overline{g}$ and $\overline{\overline{g}}$ do not exist whereas $\overline{\overline{\overline{g}}}$ does and is given explicitly by
\begin{subequations}\label{eq:gcontinuation}
\begin{align}
\overline{\overline{\overline{g}}}_1=&\frac{(1-q^3t^2)g_1+(1-q^2)g_2}{1-q^5t^2},\\
\overline{\overline{\overline{g}}}_2=&\frac{qt^2(1-q^2)g_1+(1-q^3t^2)g_2}{1-q^5t^2},\\
\overline{\overline{\overline{g}}}_3=&g_4+(q^{-2}-1)g_1g_2+q^{-4}(q^2-1)g_1^2+\frac{(1-q^2)(g_1-q^2g_2)((2-q^2)g_1-q^2g_2)}{q^4(1-q^5t^2)}\\
\nonumber &+\frac{(1-q^2)^2(g_1-q^2g_2)^2}{q^4(1-q^5t^2)^2},\\
\overline{\overline{\overline{g}}}_4=&0.
\end{align}
\end{subequations}
Similarly the inverse time-evolution is singular only when $g_4=0$, in which case the first and second inverse iterates do not exist, whereas the third one does. We say that $g(t)$ is singular at $t_0$ when it does not exist at $t=t_0$.
The continuation formulae \eqref{eq:gcontinuation} of $q\Pfour^{\text{mod}}(a)$ can be obtained by means of direct calculation.

Considering the forward iteration, $\overline{g}$ is ill-defined if and only if $g_3=0$. So let us take any $g^*\in\mathcal{G}(a)$ with $g_3^*=0$ and perturb around it within $\mathcal{G}(a)$, setting
\begin{equation*}
g_1=g_1^*+\mathcal{O}(\epsilon),\quad g_2=g_2^*+\mathcal{O}(\epsilon),\quad g_3=\epsilon+\mathcal{O}(\epsilon^2),\quad g_4=g_4^*+\mathcal{O}(\epsilon),
\end{equation*}
in particular $g=g^*+\mathcal{O}(\epsilon)$, as $\epsilon\rightarrow 0$. Then direct calculation gives
\begin{align*}
\overline{g}_1&=a_0^2a_2(qt^2-1)t^{-2}\epsilon^{-1}+\mathcal{O}(1),\\
\overline{g}_2&=-a_0^2a_2(qt^2-1)\epsilon^{-1}+\mathcal{O}(1),\\
\overline{g}_3&=a_0^4a_2^2q(qt^2-1)^2t^{-2}\epsilon^{-2}+\mathcal{O}(\epsilon^{-1}),\\
\overline{g}_4&=\mathcal{O}(\epsilon),
\end{align*}
which diverges, as $\epsilon\rightarrow 0$. Similarly
\begin{align*}
\overline{\overline{g}}_1&=-a_0^2a_2(qt^2-1)q^{-2}t^{-2}\epsilon^{-1}+\mathcal{O}(1),\\
\overline{\overline{g}}_2&=a_0^2a_2q^3(qt^2-1)\epsilon^{-1}+\mathcal{O}(1),\\
\overline{\overline{g}}_3&=\mathcal{O}(\epsilon),\\
\overline{\overline{g}}_4&=-a_0^4a_2^2q(qt^2-1)^2t^{-2}\epsilon^{-2}+\mathcal{O}(\epsilon^{-1}),
\end{align*}
which diverges, as $\epsilon\rightarrow 0$. However, upon calculating the third iteration, we find
	\begin{align*}
	\overline{\overline{\overline{g}}}_1=&\frac{(1-q^3t^2)g_1^*+(1-q^2)g_2^*}{1-q^5t^2}+\mathcal{O}(\epsilon),\\
	\overline{\overline{\overline{g}}}_2=&\frac{qt^2(1-q^2)g_1^*+(1-q^3t^2)g_2^*}{1-q^5t^2}+\mathcal{O}(\epsilon),\\
	\overline{\overline{\overline{g}}}_3=&g_4^*+(q^{-2}-1)g_1^*g_2^*+q^{-4}(q^2-1)(g_1^*)^2+\frac{(1-q^2)(g_1^*-q^2g_2^*)((2-q^2)g_1^*-q^2g_2^*)}{q^4(1-q^5t^2)}\\
	&+\frac{(1-q^2)^2(g_1^*-q^2g_2^*)^2}{q^4(1-q^5t^2)^2}+\mathcal{O}(\epsilon),\\
	\overline{\overline{\overline{g}}}_4=&\mathcal{O}(\epsilon).
	\end{align*}
which converges to \eqref{eq:gcontinuation}, as $\epsilon\rightarrow 0$. We conclude that the singularity is confined within three iterations. The singularity analysis of the inverse time evolution follows by similar arguments.

%% file: cubicderivation.tex
\section{Cubic surface calculations}
\label{append:derivationcubic}
Recall the definition of the determinants $\Delta_1$ and $\Delta_2$ in equations \eqref{eq:algebraicconstraint} and \eqref{eq:detalpha}. Each of these determinants defines a cubic equation in the variables $\rho_{1,2,3}^{x,y}$.
The aim of this section is to show that these cubics are proportional to each other
\begin{equation}\label{eq:determinantequality}
    \Delta_2=-\lambda_0^3 \Delta_1,
\end{equation}
and that they are proportional to the cubic defined in equation \eqref{eq:thom},
\begin{equation}\label{eq:determinantsimple}
    \Delta_2=\frac{1}{a_1^2a_2^2}\theta_q(a_0,a_1,a_2)\theta_q(-\lambda_0)^2T_{hom}(\rho_1^x,\rho_1^y,\rho_2^x,\rho_2^y,\rho_3^x,\rho_3^y;\lambda_0,a).
\end{equation}

Firstly, we derive equation \eqref{eq:determinantequality}. To this end, let us note that we may alternatively write the functions $v_{1,2,3}$ as
\begin{align*}
v_1(z)&=\frac{z}{\lambda_0x_1}\theta_q(z/x_2,z/x_3,-x_1/z\lambda_0),\\ v_2(z)&=\frac{z}{\lambda_0x_2}\theta_q(z/x_1,z/x_3,-x_2/z\lambda_0),\\
 v_3(z)&=\frac{z}{\lambda_0x_3}\theta_q(z/x_1,z/x_2,-x_3/z\lambda_0),
\end{align*}
due to the symmetries \eqref{eq:thetasym} of the $q$-theta function.

Using these alternative expressions for $v_{1,2,3}$ and expanding the difference of both sides of equation \eqref{eq:determinantequality}, i.e.
\begin{equation*}
    \Delta=\Delta_2+\lambda_0^3 \Delta_1,
\end{equation*}
in terms of the variables $\rho_{1,2,3}^{x,y}$, we find that all terms cancel except for
\begin{equation*}
    \Delta = \rho_1^x\rho_2^x\rho_3^x(U_1-U_2)(V_1-V_2),
\end{equation*}
where
\begin{align*}
    U_1&=\theta_q\left(-\frac{x_1}{x_2},-\frac{x_2}{x_3},-\frac{x_3}{x_1}\right),& V_1&=\theta_q\left(\frac{x_1}{x_2}\lambda,\frac{x_2}{x_3}\lambda,\frac{x_3}{x_1}\lambda\right),\\
    U_2&=\theta_q\left(-\frac{x_2}{x_1},-\frac{x_3}{x_2},-\frac{x_1}{x_3}\right),& V_2&=\theta_q\left(\frac{x_2}{x_1}\lambda,\frac{x_3}{x_2}\lambda,\frac{x_1}{x_3}\lambda\right).
\end{align*}
However, it follows directly from the symmetries \eqref{eq:thetasym} of the $q$-theta function that $U_1=U_2$ and thus $\Delta=0$, which proves equality \eqref{eq:determinantequality}.

Next, we derive equation \eqref{eq:determinantsimple}. To this end, we first expand $\Delta_2$ in terms of the variables $\rho_{1,2,3}^{x,y}$, yielding the cubic
\begin{align}
    \Delta_2=&\delta_{123}\rho_1^x\rho_2^x\rho_3^x
    +\delta_{23}\rho_1^y\rho_2^x\rho_3^x
    +\delta_{13}\rho_1^x\rho_2^y\rho_3^x
    +\delta_{12}\rho_1^x\rho_2^x\rho_3^y\label{eq:delta2}\\
    &+\delta_{1}\rho_1^x\rho_2^y\rho_3^y
    +\delta_{2}\rho_1^y\rho_2^x\rho_3^y
    +\delta_{3}\rho_1^y\rho_2^y\rho_3^x
    +\delta_{0}\rho_1^y\rho_2^y\rho_3^y,\nonumber
\end{align}
with coefficients $\delta_0,\delta_1,\delta_2$ and $\delta_3$ given by
\begin{align*}
    \delta_0&=\theta_q\left(+\frac{x_1}{x_2},+\frac{x_2}{x_1},+\frac{x_1}{x_3},+\frac{x_3}{x_1},+\frac{x_2}{x_3},+\frac{x_3}{x_2}\right)\theta_q(-\lambda_0)^3,\\
    \delta_1&=\theta_q\left(-\frac{x_1}{x_2},+\frac{x_2}{x_1},-\frac{x_1}{x_3},+\frac{x_3}{x_1},+\frac{x_2}{x_3},+\frac{x_3}{x_2}\right)\theta_q(-\lambda_0)^2\theta_q(\lambda_0),\\
    \delta_2&=\theta_q\left(+\frac{x_1}{x_2},-\frac{x_2}{x_1},+\frac{x_1}{x_3},+\frac{x_3}{x_1},-\frac{x_2}{x_3},+\frac{x_3}{x_2}\right)\theta_q(-\lambda_0)^2\theta_q(\lambda_0),\\
    \delta_3&=\theta_q\left(+\frac{x_1}{x_2},+\frac{x_2}{x_1},+\frac{x_1}{x_3},-\frac{x_3}{x_1},+\frac{x_2}{x_3},-\frac{x_3}{x_2}\right)\theta_q(-\lambda_0)^2\theta_q(\lambda_0),
\end{align*}
$\delta_{23},\delta_{13}$ and $\delta_{12}$ given by
\begin{align*}
    \delta_{23}&=\theta_q\left(+\frac{x_1}{x_2},-\frac{x_2}{x_1},+\frac{x_1}{x_3},-\frac{x_3}{x_1},-\lambda_0\right)h(\lambda_0;x_2,x_3),\\
    \delta_{13}&=\theta_q\left(+\frac{x_1}{x_2},-\frac{x_2}{x_1},+\frac{x_2}{x_3},-\frac{x_3}{x_2},-\lambda_0\right)h(\lambda_0;x_1,x_3),\\
    \delta_{12}&=\theta_q\left(-\frac{x_1}{x_3},+\frac{x_3}{x_1},+\frac{x_3}{x_2},-\frac{x_2}{x_3},-\lambda_0\right)h(\lambda_0;x_1,x_2),
\end{align*}
where
\begin{equation*}
    h(\lambda;z_1,z_2):=\theta_q(-\frac{z_1}{z_2},-\frac{z_2}{z_1})\theta_q(\lambda)^2-\theta_q(-1)^2\theta_q(-\frac{z_1}{z_2}\lambda,-\frac{z_2}{z_1}\lambda),
\end{equation*}
and the coefficient $\delta_{123}$ given by
\begin{align*}
    \delta_{123}=&H(\lambda_0),
    \end{align*}
    where
 \begin{align*}
   H(\lambda):=&
    +\theta_q\left(-\frac{x_1}{x_2},-\frac{x_2}{x_1},-\frac{x_1}{x_3},-\frac{x_3}{x_1},-\frac{x_2}{x_3},-\frac{x_3}{x_2}\right)\theta_q(\lambda)^3\\
    &-\theta_q\left(\lambda\hspace{0.7mm}\frac{x_1}{x_2},\lambda\hspace{0.7mm}\frac{x_2}{x_1},-\frac{x_1}{x_3},-\frac{x_3}{x_1},-\frac{x_2}{x_3},-\frac{x_3}{x_2}\right)\theta_q(\lambda)\theta_q(-1)^2\\
   & -\theta_q\left(-\frac{x_1}{x_2},-\frac{x_2}{x_1},\lambda\hspace{0.7mm}\frac{x_1}{x_3},\lambda\hspace{0.7mm}\frac{x_3}{x_1},-\frac{x_2}{x_3},-\frac{x_3}{x_2}\right)\theta_q(\lambda)\theta_q(-1)^2\\
   & -\theta_q\left(-\frac{x_1}{x_2},-\frac{x_2}{x_1},-\frac{x_1}{x_3},-\frac{x_3}{x_1},\lambda\hspace{0.7mm}\frac{x_2}{x_3},\lambda\hspace{0.7mm}\frac{x_3}{x_2}\right)\theta_q(\lambda)\theta_q(-1)^2\\
   & +\theta_q\left(\lambda\hspace{0.7mm}\frac{x_1}{x_2},-\frac{x_2}{x_1},-\frac{x_1}{x_3},\lambda\hspace{0.7mm}\frac{x_3}{x_1},\lambda\hspace{0.7mm}\frac{x_2}{x_3},-\frac{x_3}{x_2}\right)\theta_q(-1)^3\\
  &  +\theta_q\left(-\frac{x_1}{x_2},\lambda\hspace{0.7mm}\frac{x_2}{x_1},\lambda\hspace{0.7mm}\frac{x_1}{x_3},-\frac{x_3}{x_1},-\frac{x_2}{x_3},\lambda\hspace{0.7mm}\frac{x_3}{x_2}\right)\theta_q(-1)^3.
\end{align*}

In order to derive equation \eqref{eq:determinantsimple}, we factorise the functions $h(\lambda;z_1,z_2)$ and $H(\lambda)$ into simple factors of $q$-theta functions. We start with $h(\lambda;z_1,z_2)$. Note that this function is an element of the vector space $V_2(1)$, namely, it is analytic in $\lambda$ on $\mathbb{C}^*$ and satisfies
\begin{equation*}
    h(q\lambda;z_1,z_2)=\lambda^{-2}h(\lambda;z_1,z_2).
\end{equation*}
Furthermore, it is easy to see that $h(-1;z_1,z_2)=0$, hence, by Lemma \ref{lem:classificationtheta},
\begin{equation}\label{eq:hexplicitpre}
    h(\lambda;z_1,z_2)=c(z_1,z_2)\theta_q(-\lambda)^2,
\end{equation}
for some yet to be determined coefficient $c(z_1,z_2)$ which is independent of $\lambda$. Then, by  evaluating equation \eqref{eq:hexplicitpre} at $\lambda=1$, we obtain
\begin{equation*}
    c(z_1,z_2)=-\theta_q\left(\frac{z_1}{z_2},\frac{z_2}{z_1}\right)
\end{equation*}
and thus
\begin{equation}\label{eq:hexplicit}
    h(\lambda;z_1,z_2)=-\theta_q\left(\frac{z_1}{z_2},\frac{z_2}{z_1}\right)\theta_q(-\lambda)^2.
\end{equation}

We follow the same procedure to compute a factorisation of $H(\lambda)$. We note that it is an element of the vector space $V_3(-1)$, i.e. it is analytic on $\mathbb{C}^*$ and
\begin{equation*}
    H(q\lambda)=-\lambda^{-3}H(\lambda).
\end{equation*}
Furthermore, it is easy to check by direct calculation that $H(1)=H(-1)=0$, and hence
\begin{equation}\label{eq:Hexplicitpre}
    H(\lambda)=c_H\theta_q(\lambda)\theta_q(-\lambda)^2,
\end{equation}
for some yet to be determined coefficient $c_H$, which is independent of $\lambda$. To determine this constant, we evaluate both sides of equation \eqref{eq:Hexplicitpre} at $\lambda=x_1$. By means of analogous calculations as above, we can simplify $H(x_1)$ to obtain
\begin{equation*}
    H(x_1)=\frac{\theta_q(x_1)^3\theta_q(-x_1)^2\theta_q(x_2)^2\theta_q(\frac{x_2}{x_1})^2}{q^2x_1^2x_2^4},
\end{equation*}
leading to
\begin{equation}\label{eq:Hexplicit}
    H(\lambda)=\frac{\theta_q(x_1)^2\theta_q(x_2)^2\theta_q(\frac{x_2}{x_1})^2}{q^2x_1^2x_2^4}\theta_q(\lambda)\theta_q(-\lambda)^2.
\end{equation}
Now that we have explicit factorisations of the functions $h(\lambda;z_1,z_2)$ and $H(\lambda)$ (respectively given by equations \eqref{eq:hexplicit} and \eqref{eq:Hexplicit}), and thus factorised expressions for the coefficients of the cubic \eqref{eq:delta2}, Equation \eqref{eq:determinantsimple} follows immediately, upon using the symmetries \eqref{eq:thetasym} of the $q$-theta function.

%% file: main.bbl
\begin{bibdiv}
  \begin{biblist}
\bib{adams1928}{article}{
 author = {Adams, C.R.},
 journal = {Annals of Mathematics},
 number = {1/4},
 pages = {195-205},
 title = {On the {L}inear {O}rdinary q-{D}ifference {E}quation},
 volume = {30},
 year = {1928}
}
    
\bib{birkhoffgeneralized1913}{article}{
  author={Birkhoff, G.D.},
  title={The generalized {R}iemann problem for linear differential
    equations and the allied problems for linear difference and
    $q$-difference equations},
  journal= {Proceedings of the American Academy of Arts and Sciences},
  volume={49},
  pages={521--568},
  year={1913}
  }

  \bib{Borodin2005}{article}{
    author = {Borodin, A.},
    title = {Isomonodromy transformations of linear systems of difference equations},
    Journal = {Ann. Math. (2)},
    Volume = {160},
    Pages = {1141--1182},
    Year = {2005}
  }
  
  \bib{buckinghampiv}{article}{
      author = {Buckingham, R.},
    title = {Large-Degree Asymptotics of Rational Painlevé-IV Functions Associated to Generalized Hermite Polynomials},
    journal = {International Mathematics Research Notices},
    year = {2018},
}

\bib{carmichael1912}{article}{
  author= {Carmichael,R.D.},
  title={The general theory of linear $q$-difference equations},
  journal={ American Journal of Mathematics},
  volume={34},
  pages={147--168},
  year={1912}
  }

  \bib{chekhovmazzoccorubtsov2020}{article}{
  author={Chekhov, L.},
  author={Mazzocco, M.},
  author={V. Rubtsov},
  title={Quantised {P}ainlev\'e monodromy manifolds, {S}klyanin and {C}alabi-{Y}au algebras},
journal={Advances in Mathematics},
volume={376},
pages={107442},
year={2021}
}

  \bib{chekhovmazzoccorubtsov2017}{article}{
  author={Chekhov, L.},
  author={Mazzocco, M.},
  author={V. Rubtsov},
  title={{P}ainlev\'e monodromy manifolds, decorated character varieties, and cluster algebras},
journal={Int. Math. Res. Not. IMRN},
number={(24)},
pages={7639--7691},
year={2017}
}


\bib{deiftorthogonal}{book}{
  author={Deift, P.~A.},
  title={ Orthogonal polynomials and random matrices: a
    {R}iemann-{H}ilbert approach},
  volume={3},
  series={Courant Lecture Notes in Mathematics},
  publisher={New York University, Courant Institute of Mathematical Sciences, New York and American Mathematical Society, Providence, RI},
  year={1999}
  }

\bib{deift93}{article}{
  author= {Deift, P.~A.},
  author= {Zhou, X.},
  title={A steepest descent method for oscillatory {R}iemann-{H}ilbert
    problems. {A}symptotics for the {MK}d{V} equation},
  journal= {Ann. of Math. (2)},
  volume={137},
  pages={295--368},
  year={1993}
}

\bib{fokas}{book}{
  author={Fokas,A.},
  author={Its, A.},
  author={Kapaev, A.},
  author={Novokshenov, A.},
  title={Painlev\'e transcendents, the Riemann-Hilbert approach},
  volume={ 128},
  series={Mathematical Surveys and Monographs},
  publisher={American Mathematical Society},
  year={2006}
  }

\bib{fokas1991discrete}{article}{
  author={Fokas,A.},
 author={Its, A.},
  author={Kitaev, A.},
   title={Discrete {P}ainlev{\'e} equations and their appearance in quantum
  gravity},
  journal={Communications in Mathematical Physics},
  volume={142},
  pages={313--344},
  year={1991}
}

\bib{fokas92}{article}{
  author= {Fokas,A.},
  author= {Mu{\u{g}}an, U.},
  author= {Zhou, X.},
  title= {On the solvability of {P}ainlev\'e {${\rm I},\;{\rm III}$}
          and {${\rm V}$}},
 journal= {Inverse Problems},
 volume={8},
 pages={757--785},
 year={1992}
}

\bib{ggkm67}{article}{
  author={Gardner, C.S.},
  author={Greene, J.~M.},
  author={Kruskal, M.~D.},
  author={Miura, R.~M.},
  title={Method for solving the {K}orteweg-de{V}ries equation},
  journal={Phys. Rev. Lett.},
  volume={19},
  pages={1095--1097},
  year={1967}
}

\bib{IKF95}{article}{
  author={Its, A.R.},
  author={Kitaev,A.V.},
  author={Fokas,A.S.},
  title={Matrix models of two-dimensional quantum gravity and isomonodromic solutions of ``discrete {P}ainlev{\'e}'' equations},
  journal={Journal of Mathematical Sciences},
  volume={73},
  pages={415--429},
  year={1995}
}


\bib{its90}{article}{
  author={Its, A.},
  author={Kitaev, A.},
  author={Fokas,A.},
  title={An isomonodromy approach to the theory of
    two-dimensional quantum
  gravity},
  journal={Uspekhi Mat. Nauk},
  volume={45},
  pages={135--136},
  year={1990}
}


\bib{JMMS}{article}{
  author={Jimbo, M.},
  author={Miwa, T.},
  author={Mori, T.},
  author={Sato, M.},
  title={Density matrix of an impenetrable Bose gas and the fifth {P}ainlev\'e transcenden},
    journal={Phys. D},
    volume={1},
     pages={80--158},
     year={1980}
   }
   
\bib{jimbomiwaII}{article}{
  author={Jimbo, M.},
  author={Miwa, T.},
  title={Monodromy preserving deformation of
      linear ordinary differential
      equations with rational coefficients. {II}},
    journal={Phys. D},
    volume={2},
     pages={407--448},
     year={1981}
     }

\bib{jimbomiwaIII}{article}{
  author={Jimbo, M.},
  author={Miwa, T.},
  title={Monodromy preserving deformation of
    linear ordinary differential
  equations with rational coefficients. {III}},
    journal={Phys. D},
    volume={4},
    pages={26--46},
    year={1981/82}
    }

\bib{jimbomiwaI}{article}{
  author={Jimbo, M.},
  author={Miwa, T.},
  author={Ueno, K.},
 title={Monodromy preserving deformation of linear ordinary differential
  equations with rational coefficients. {I}. {G}eneral theory and {$\tau
  $}-function},
    journal={Phys. D},
    volume={2},
    pages={306--352},
         year={1981}
     }
     

  \bib{jimbonagoyasakai}{article}{
  author={Jimbo, M.},
  author={Nagoya, H.},
  author={Sakai, H.},
  title={CFT approach to the q-Painlev\'e VI equation},
  journal={Journal of Integrable Systems},
  volume={2},
  year={2017}
  }

  

\bib{jimbosakai}{article}{
  author={Jimbo, M.},
  author={Sakai, H.},
  title={A $q$-analogue of the sixth {P}ainlev{\'e} equation},
  journal={ Lett. Math. Phys.},
  volume={38},
  pages={145--154},
  year={1996}
}



\bib{joshicbms2019}{book}{
  author={Joshi, N.},
  title={Discrete Painlev\'e Equations},
  series={ CBMS Regional Conference Series in Mathematics},
  volume={131},
  date={2019},
  publisher={American Mathematical Society},
  address={Providence, Rhode Island}
}

\bib{joshinobu2016}{article}{
  author={Joshi, N.},
  author={Nakazono, N.},
  title={Lax pairs of discrete {P}ainlev\'e equations:
    {$(A_2+A_1)^{(1)}$} case},
  journal={ Proc. Roy Soc. A.},
  volume={472},
  pages={20160696},
  year={2016}
  }

\bib{joshiroffelsen}{article}{
   author={Joshi, N.},
   author={Roffelsen, P.},
   title={Analytic solutions of $q$-$P(A_1)$ near its critical points},
   journal={Nonlinearity},
   volume={29},
   pages={3696},
    year={2016}
    }

\bib{kny:17}{article}{
  author={Kajiwara, K.},
  author={ Noumi, M.},
  author={Yamada, Y.},
  title={Geometric aspects of painlev{\'e} equations},
  journal={J. Phys. A Math. and Theor.l},
  volume={50},
  pages={073001},
  date={2017}
}

\bib{kapaev1996}{webpage}{
  author={Kapaev, A.A.},
  title={Global asymptotics of the fourth {P}ainlev\'e transcendent},
  date={1996},
  url={ftp://www.pdmi.ras.ru/pub/publicat/preprint/1996/06-96.ps.gz}
}


\bib{kapaev1998}{article}{
Author = {A.A. Kapaev},
Title = {Connection formulae for degenerated asymptotic solutions of the fourth Painlev\'e equation},
Year = {1998},
journal = {arXiv:solv-int/9805011[nlin.SI]}
}


\bib{mano2010}{article}{
  author={Mano, T.},
  title={Asymptotic behaviour around a boundary
    point of the $q$-{P}ainlev{\'e} {VI} equation
    and its connection problem},
  journal={Nonlinearity},
  volume={23},
  pages={1585--1608},
  year={2010}
}

\bib{masoeroroffelsen}{article}{
  author = {Masoero, D.},
  author=  {Roffelsen, P.},
    title = {Poles of Painlev\'e IV Rationals and their Distribution},
   JOURNAL = {SIGMA Symmetry Integrability Geom. Methods Appl.},
    VOLUME = {14},
      YEAR = {2018},
     PAGES = {Paper No. 002, 49},
}

\bib{masoeroroffelsen2}{article}{
     author = {Masoero, D.},
  author=  {Roffelsen, P.},
    title = {Roots of generalised Hermite polynomials when both parameters are large},
  journal = {ArXiv e-prints},
   eprint = {1907.08552},
     year = {2019}}


  \bib{murata04}{article}{
  author={Murata, M.},
  title={Lax forms of the $q$-{P}ainlev{\'e} equations},
  journal={J. Phys. A: Math. Theor.},
  volume={42},
  pages={115201 (17pp)},
  year={2009}
  }

  
\bib{ormerodconnection}{article}{
 author={Ormerod, C.M.},
  author={Witte, N.S.},
  author={Forrester, P.J.},
  title={Connection preserving deformations and q-semi-classical orthogonal polynomials},
  journal={Nonlinearity},
  volume={24},
  pages={2405},
  date={2011}}


  \bib{rains2011}{article}{
    author = {Rains, E.M.},
     title = {An isomonodromy interpretation of the hypergeometric solution
              of the elliptic {P}ainlev\'e equation (and generalizations)},
   journal = {SIGMA Symmetry Integrability Geom. Methods Appl.},
    volume = {7},
      year = {2011},
     pages = {Paper 088, 24}
   }

   
\bib{sauloy}{article}{
  author={Ramis, J.P.},
  author={Sauloy, J.},
  author={Zhang,C.},
  title={{L}ocal analytic classification of $q$-difference
       equations},
     journal={Ast\'erisque},
     volume={355},
       year={2013}
       }


\bib{phdroffelsen}{misc}{
  author={Roffelsen, P.},
  title={On the global asymptotic analysis of a $q$-discrete
        {P}ainlev{\'e} equation},
    type={PhD thesis},
    organization={The University of Sydney},
    note = {\url{https://ses.library.usyd.edu.au/handle/2123/16601}},
    year={2017}
    }

\bib{sakai2001}{article}{
    AUTHOR = {Sakai, H.},
     TITLE = {Rational surfaces associated with affine root systems and
              geometry of the {P}ainlev\'e equations},
   JOURNAL = {Comm. Math. Phys.},
    VOLUME = {220},
      YEAR = {2001},
    NUMBER = {1},
     PAGES = {165--229}
}


\bib{sauloy2002}{collection.article}{
  author={Sauloy, J.},
  title={Galois theory of {$q$}-difference equations: the ``analytical''
  approach},
  booktitle= {Differential equations and the {S}tokes phenomenon},
  pages={277--292},
  publisher={World Sci. Publ., River Edge, NJ},
  year={2002}
}

\bib{shohat1939}{article}{
     author={Shohat, Jacques},
     title={A differential equation for orthogonal polynomials},
     journal={Duke Mathematical Journal},
     volume={5},
     pages={401--417},
     year={1939}
}

\bib{putgalois1997}{collection.article}{
  author={van~der Put, M.},
  author={Singer,M.F.},
  title={Galois theory of difference equations},
  volume={1666},
  series={Lecture Notes in Mathematics},
  publisher={Springer-Verlag, Berlin},
  year={1997}
  }

\bib{putsaito2009}{article}{
  author={van~der Put, M.},
  author={Saito, M.},
  title={Moduli spaces for linear
    differential equations and the {P}ainlev\'e
  equations},
journal={Annales de l'Institut Fourier},
volume={59},
pages={2611--2667},
year={2009}
}



\end{biblist}
\end{bibdiv}
